\documentclass[10pt, final, journal, letterpaper, twocolumn]{IEEEtran}

\usepackage[dvips]{graphicx}
\usepackage{times}
\usepackage{cite}
\usepackage{amsmath}
\usepackage{array}
\usepackage{amssymb}

\usepackage{stfloats}
\usepackage{rotating,threeparttable,booktabs}
\usepackage{bm}
\usepackage{dcolumn,booktabs}
\usepackage{multirow}
\usepackage{graphicx}
\usepackage{subfigure}
\usepackage{color}
\usepackage{multicol}

\usepackage{amsmath,amsthm,amssymb,amsfonts}
\makeatletter
\thm@headfont{\sc}
\makeatother
\newtheorem{theorem}{Theorem}
\newtheorem{property}{Property}

\newtheorem{lemma}{Lemma}

\newtheorem{prop}{Proposition}

\newtheorem{corollary}{Corollary}

\begin{document}

\title{MIMO Multiway Relaying with Pairwise Data Exchange: A Degrees of Freedom Perspective}
\author{\authorblockA {Rui Wang, \emph{Member}, \emph{IEEE} and Xiaojun Yuan, \emph{Member}, \emph{IEEE}}
\thanks{R. Wang is with the Institute of Network Coding, The Chinese University of
Hong Kong, Hong Kong, Email: ruiwang@inc.cuhk.edu.hk.}
\thanks{X. Yuan is with School of Information Science and Technology, ShanghaiTech University, Shanghai, 200031, China, Email: yuanxj@shanghaitech.edu.cn.}}
\maketitle

\begin{abstract}
In this paper, we study achievable degrees of freedom (DoF) of a multiple-input multiple-output (MIMO) multiway relay channel (mRC) where $K$ users, each equipped with $M$ antennas, exchange messages in a pairwise manner via a common $N$-antenna relay node.
A novel and systematic way of joint beamforming design at the users and at the relay is proposed to align signals for efficient implementation of physical-layer network coding (PNC).
It is shown that, when the user number $K=3$, the proposed beamforming design can achieve the DoF capacity of the considered mRC for any $(M,N)$ setups.
For the scenarios with $K>3$,
we show that the proposed signaling scheme can be improved by disabling a portion of relay antennas so as to align signals more efficiently. Our analysis reveals that the obtained achievable DoF is always piecewise linear, and is bounded either by the number of user antennas $M$ or by the number of relay antennas $N$.
Further, we show that the DoF capacity can be achieved for $\frac{M}{N} \in \left(0,\frac{K-1}{K(K-2)} \right]$ and $\frac{M}{N} \in \left[\frac{1}{K(K-1)}+\frac{1}{2},\infty \right)$, which
provides a broader range of the DoF capacity than the existing results.
Asymptotic DoF as $K\rightarrow \infty$ is also derived based on the proposed signaling scheme.
\end{abstract}

\begin{IEEEkeywords}
Multiple-input multiple-output (MIMO), physical-layer network coding (PNC), multiway relay channel (mRC), signal space alignment.
\end{IEEEkeywords}

\section{Introduction}

Recently, an exponential increase in the demands of wireless service has imposed a significant challenge on the design of wireless networks. Advanced techniques, such as physical-layer network coding (PNC), has been developed to achieve high spectrum efficiency \cite{ZhangMobicom06,Petar07}.
The simplest model for PNC is the two-way relay channel (TWRC) where two users exchange messages with the help of a relay node. With the well-known two-phase PNC protocol, the relay node receives a combination of the signals transmitted from the two users in the first phase, and then broadcasts a network-coded message in the second phase. The desired message is then extracted at each user end by exploiting the knowledge of the self-message. As compared to conventional one-way relaying where four phases are required in one round of information exchange,
PNC potentially achieves $100\%$ improvement in spectrum efficiency over TWRCs.

Abundant progresses have been made on the PNC design for TWRCs; see \cite{ZhangJSAC09,NamIT10,HJYangIT11,YangIT11,YuanIT13} and the references therein. In particular, it was shown in \cite{NamIT10} that, with nested lattice coding, the capacity of the TWRC can be achieved within $\frac{1}{2}$ bit. Later, the authors in \cite{HJYangIT11,YangIT11,YuanIT13} introduced the multiple-input multiple-output (MIMO) technique into TWRCs. It was revealed that the space-division based network coding scheme proposed in \cite{YuanIT13} achieves the asymptotic capacity of the MIMO TRWC at high signal-to-noise ratio (SNR) within $\frac{1}{2}\log(\frac{5}{4})$ bit per relay spatial dimension for an arbitrary antenna configuration.

A natural generalization of the TWRC is the multiway relay channel (mRC), where multiple users exchange messages with the help of a single relay.
Several mRC models have been studied in the literature recently.
Specifically, the authors in \cite{Jang_CL2010} studied a cellular two-way relaying model
where a base station exchanges private messages with multiple mobile users via a relay node;
the authors in \cite{Sezgin_ISIT2009,Chen_TWC2009} investigated mRCs in which the users are grouped into pairs and the two users in each pair exchange information with each other;
more generally, the authors in \cite{gunduz51} studied clustered mRCs, in which the users in the network are grouped into clusters and each user in a cluster wants to exchange information with the other users in the same cluster.
Approximate capacities of these mRC models were studied in \cite{gunduz51} and \cite{Chaaban_IT2013}, while the exact capacity characterizations still remain open.
Also, these initial works on mRC are limited to the single-antenna setup, i.e., each node in the network is equipped with one antenna.

The MIMO technique has been introduced into mRCs to allow spatial multiplexing. In a MIMO mRC, as each user in general transmits multiple spatial streams, a new challenge to be addressed is to mitigate the inter-stream interference at the relay and at the user ends.
Degrees of freedom (DoF) is a critical metric in characterizing the fundamental capacity of a wireless network \cite{Jafar08,Wang_IT2011}. The DoF of the MIMO mRC
has been previously studied in \cite{Lee_TWC2013, Lee10,Lee12,Chaaban_ISIT2013,Wang_IT2014,Tian_IT2013,Yuan_IWC2013}. For example,
the authors in \cite{Lee10} investigated the DoF capacity of the MIMO Y channel (a special case of the MIMO mRC with three users)
and showed that the DoF capacity can be achieved when $\frac{M}{N}\geq \frac{2}{3}$, where $M$ denotes the number of antennas at each user, and $N$ denotes the number of antennas at the relay.
The authors in \cite{Chaaban_ISIT2013} generalized the results in \cite{Lee10} by considering a three-user asymmetric MIMO Y channel with different numbers of antennas at the users, and proved that the DoF capacity can be achieved for arbitrary antenna setups. Recently, as parallel to the work in this paper, the work in \cite{Wang_IT2014} established the DoF capacity of the four-user symmetric MIMO Y channel for arbitrary antenna setups.
Further, the authors
in \cite{Tian_IT2013,Yuan_IWC2013}
studied more general data exchange models in which the users in the network are grouped into clusters, and each user in a cluster exchanges information only with the other users in the same cluster. In particular, the authors in \cite{Tian_IT2013} derived sufficient conditions on the antenna configuration to achieve the DoF capacity of a clustered mRC with pairwise data exchange, in which each user in a cluster sends a different message to each of the other users in the same cluster. Note that the data exchange models considered in \cite{Lee10} and \cite{Lee12} can be regarded as
the one-cluster case of the model studied in \cite{Tian_IT2013}.
Moreover, the author in \cite{Yuan_IWC2013} derived an achievable DoF for a clustered MIMO mRC with full data exchange, i.e., each user in a cluster delivers a common message to all the other users in the same cluster.

In this work, we study a symmetric MIMO mRC with pairwise data exchange, and derive an achievable DoF for an arbitrary setup of the antenna numbers $(M,N)$ and the user number $K$.
Roughly speaking, the DoF of a network is the number of independent spatial streams that can be supported by the network.
In the MIMO mRC of concern, multiple users are simultaneously served by a common relay.
To ensure that multiple spatial streams are still separable at every user end,
the number of relay antennas is usually the bottleneck of the network to achieve a higher DoF.
Therefore, the challenge is how to align the user and relay signals to efficiently utilize the relay's signal space.
To this end, we propose a novel and systematic beamforming design to achieve efficient signal alignment.
Specifically, we refer to a bunch of spatial streams as a \emph{unit}, in which each pair of users who want to exchange information contribute two spatial streams, one from each user; each spatial stream impinges upon (or is emitted from) the relay's antenna array at a certain direction, and these directions
form a spatial structure, referred to as a \emph{pattern}. The dimension of the space spanned by the spatial streams in a pattern gives a metric to evaluate the efficiency of this pattern. Then, the signal alignment problem is to construct units with the most efficient patterns to occupy the overall relay's signal space. An achievable DoF can be obtained by counting the number of units that can be constructed for any given antenna setup of $(M,N)$.

The main contributions of this work are summarized as follows:
\begin{itemize}
  \item We show that, for the considered MIMO mRC with $K=3$, the proposed signal alignment scheme achieves the DoF capacity for an arbitrary  $(M,N)$ setup, which coincides with the DoF result of \cite{Chaaban_ISIT2013}, and improves the existing DoF capacity result in \cite{Lee10} by including $\frac{M}{N}\in \left(0,\frac{2}{3}\right)$.
  \item For the case of $K>3$, we derive the DoF capacity of the MIMO mRC for
   $\frac{M}{N} \in \left(0,\frac{K-1}{K(K-2)} \right]$ and $\frac{M}{N} \in \left[\frac{1}{K(K-1)}+\frac{1}{2},\infty \right)$. This result is stronger than the previous result obtained in \cite{Tian_IT2013}, where the achievability of the DoF capacity is limited in the ranges of $\frac{M}{N} \in \left(0,\frac{1}{K} \right]$ and $\frac{M}{N} \in \left[\frac{1}{K(K-1)}+\frac{1}{2},\infty \right)$.
  \item
   For $K>3$, we also derive an achievable DoF for an arbitrary setup of antenna numbers $(M,N)$ satisfying $\frac{M}{N} \in \left(\frac{K-1}{K(K-2)}, \frac{1}{K(K-1)}+\frac{1}{2} \right)$. Our analysis reveals that the achievable DoF is piecewise linear
   and is bounded either by the number of antennas at each user or by the number of antennas at the relay. This piecewise linearity is similar to the DoF capacity obtained for the MIMO interference channel in \cite{Wang_IT2011}.
  \item Finally, we derive an asymptotic achievable DoF when $K$ tends to infinity. We show that the derived achievable total DoF is upper bounded by $\min\left(\frac{N^2}{N-M},2N\right)$ for arbitrary values of $M$ and $N$.
\end{itemize}

The rest of the paper is organized as follows. In Section II, we present the system model.
In Section III, a DoF upper bound is introduced as the benchmark of the system design.
The DoF capacity of the considered MIMO mRC with three users is presented in Section IV.
In Section V, we generalize the results to the case of an arbitrary number of users.
In Section VI, an improved DoF result is presented by disabling a portion of relay antennas.
Finally, we conclude the paper in Section VII.

\emph{Notation}:
Scalars, vectors, and  matrices are denoted by lowercase
regular letters, lowercase bold letters,
and uppercase bold letters, respectively. For a matrix $\mathbf{A}$, $\mathbf{A}^T$ and $\mathbf{A}^H$
denote the transpose and the Hermitian transpose of $\bf A$, respectively;
${\rm tr}({\bf A})$ and ${\bf A}^{-1}$ stand for the trace and the inverse of ${\bf A}$, respectively;
${\rm diag}({\bf A}_1,{\bf A}_2,\cdots,{\bf A}_n)$ denotes a block-diagonal matrix with the $i$-th diagonal block specified by ${\bf A}_i$ where $n$ is an integer;
$\mathrm{span}(\mathbf{A})$ and $\mathrm{null}(\mathbf{A})$ denote
the column space and  the nullspace of $\bf A$, respectively;
${\bf I}_n$ denotes an $n \times n$ identity matrix;
$\mathrm{dim}(\mathcal{S})$ denotes the dimension of a space $\mathcal{S}$;
$\mathcal{S}\cap \mathcal{U}$ and $\mathcal{S}\oplus \mathcal{U}$ denote the intersection and the direct sum of two spaces $\mathcal{S}$ and $\mathcal{U}$, respectively; $\mathcal{\mathbb{R}}^{n\times m}$ and $\mathcal{
\mathbb{C}}^{n\times m}$ denote the $n \times m$ dimensional real space and complex
space, respectively; $\log (\cdot )$ denotes the logarithm
with base $2$;
$[\cdot ]^{+}$ denotes $\max \{\cdot ,0\}$;
$\mathcal{CN}(\mu ,\sigma ^{2})$ denotes the distribution of a circularly symmetric complex Gaussian random variable with mean $\mu$ and variance $\sigma ^{2}$; $\left(
                                                                               \begin{smallmatrix}
                                                                                 n \\
                                                                                 m \\
                                                                               \end{smallmatrix}
                                                                             \right)=\frac{n!}{m!(n-m)!}$ denotes the binomial coefficient indexed by $n$ and $m$.

\section{System Model}

\begin{figure}[tp]
\begin{centering}
\includegraphics[scale=0.8]{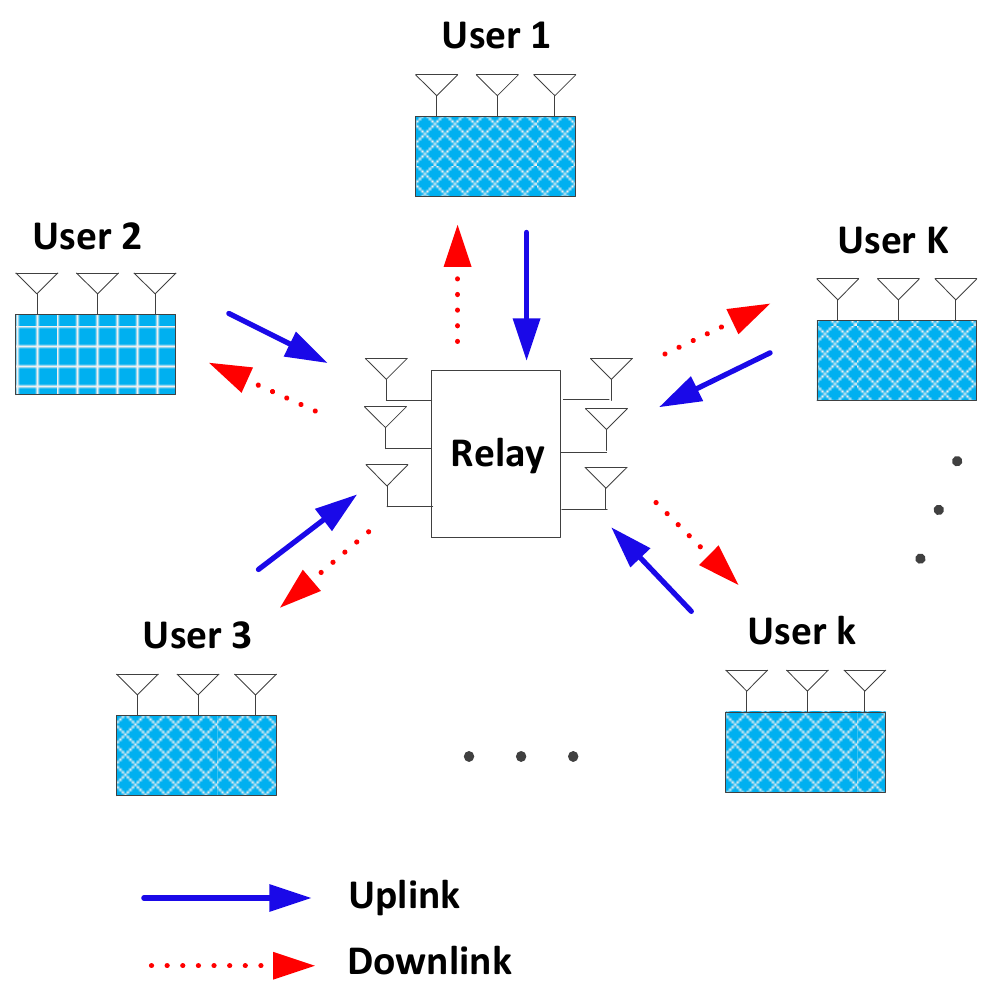}
\vspace{-0.1cm}
\caption{An illustration of the MIMO mRC with $K$ users operating in pairwise exchange.} \label{Fig_Config_MIMOTWRC1}
\end{centering}
\vspace{-0.3cm}
\end{figure}

\subsection{Channel Model}
Consider a discrete memoryless symmetric MIMO mRC $(M,N,K)$, where $K$ users, each equipped with $M$ antennas, exchange messages in a pairwise manner with the help of a common $N$-antenna relay node, as illustrated in Fig.~\ref{Fig_Config_MIMOTWRC1}.
Full-duplex communication is assumed, i.e., all the nodes transmit and receive signal simultaneously.\footnote{All the DoF results obtained in this paper directly hold for the half-duplex case by including a multiplicative factor of $\frac{1}{2}$.}
The direct links between users are ignored due to physical impairments such as shadowing and path loss of wireless fading channels.

Each round of information exchange is implemented in two phases with equal time duration $T$. In the first phase (termed the \emph{uplink phase}), all the users simultaneously transmit signals to the relay node. The received signal at the relay node can be written as
\begin{equation}\label{System-1}
\mathbf{Y}_R = \sum_{k=1}^{K}\mathbf{H}_{k}\mathbf{X}_{k}+\mathbf{Z}_R,~k \in \mathcal{I}_K \triangleq \{1,2,\cdots,K\}
\end{equation}
where $\mathbf{H}_{k}\in \mathcal{\mathbb{C}}^{N\times M}$ denotes the channel matrix from user $k$ to the relay; $\mathbf{X}_{k}
\in \mathcal{\mathbb{C}}^{M\times T}$ is the transmit signal from user $k$;
similarly, $\mathbf{Y}_{R}
\in \mathcal{\mathbb{C}}^{N\times T}$ denotes the received signal at the relay node; $\mathbf{Z}_R \in \mathcal{\mathbb{C}}^{N\times T}$ is the additive white Gaussian noise (AWGN) matrix at the relay node and the elements are independently drawn from the distribution of $\mathcal{CN}(0, \sigma^2_R)$. The transmit signal $\mathbf{X}_{k}$ at user $k$ satisfies the power constraint of
\begin{equation}\label{System-3} \nonumber
\frac{1}{T} \mathrm{tr}(\mathbf{X}_{k}\mathbf{X}_{k}^H) \leq P_k, \ \ k \in \mathcal{I}_K
\end{equation}
where $P_k$ is the maximum transmission power allowed at user $k$.

In the second phase (termed the \emph{downlink phase}), the relay sends the processed signals to all user ends. The received signal at user $k$ is denoted by
\begin{equation}\label{System-2}
\mathbf{Y}_{k} = \mathbf{G}_{k}\mathbf{X}_{R}+\mathbf{Z}_{k},  \ \ k \in \mathcal{I}_K
\end{equation}
where $\mathbf{G}_{k}\in \mathcal{\mathbb{C}}^{M\times N}$ denotes the channel matrix from the relay to user $k$;
$\mathbf{X}_{R}\in \mathcal{\mathbb{C}}^{N\times T}$ is the transmit signal at the relay node; $\mathbf{Z}_{k} \in \mathcal{\mathbb{C}}^{N\times T}$ is the AWGN noise matrix at user $k$ with the elements independently drawn from $\mathcal{CN}(0,\sigma^2_k)$. The transmit signal $\mathbf{X}_{R}$ satisfies the power constraint of
\begin{equation}\label{System-4} \nonumber
\frac{1}{T} \mathrm{tr}(\mathbf{X}_{R}\mathbf{X}_{R}^H) \leq P_R,
\end{equation}
where $P_R$ is the maximum transmission power allowed at the relay.

We assume that the elements of the channel matrices $\mathbf{H}_{k}$ and $\mathbf{G}_{k}$, $\forall k$, are draw from a continuous distribution, which implies that these channel matrices are of full column or row rank, whichever is smaller, with probability one. The channel state information is assumed to be perfectly known at all nodes, following the convention in \cite{Lee_TWC2013, Lee10,Lee12,Chaaban_ISIT2013,Wang_IT2014,Tian_IT2013,Yuan_IWC2013}.
It is worth noting that the considered MIMO mRC reduces to the MIMO two-way relay channel (TWRC) when $K=2$, and to the MIMO Y channel when $K=3$. As the DoF capacity of the MIMO TWRC is well understood, we henceforth focus on the case of $K\geq3$.

\subsection{Linear Signaling Scheme}
In the considered mRC, each user $k$, $k \in \mathcal{I}_K$, intends to send a private message $W^{(k,k^\prime)}$ to user $k^\prime$, $\forall k^\prime \neq k$.
The message $W^{(k,k^\prime)}$ is then encoded as $f(W^{(k,k^\prime)})=\{{\bf s}^{(k,k^\prime)}_{1}, {\bf s}^{(k,k^\prime)}_2,\cdots,{\bf s}^{(k,k^\prime)}_L\}$, where $f(\cdot)$ is an encoding function; ${\bf s}^{(k,k^\prime)}_l \in \mathbb{C}^{1 \times T}$ denotes the spatial stream transmitted in unit $l$; $L$ is the number of the units which can be supported by the network.
The goal of this work is to analyze the achievable DoF of the considered MIMO mRC.
Linear processing is assumed to be applied at the transmitter, relay, and receiver sides.
The transmit signal at user $k$ is denoted as
\begin{equation}\label{Sys_model_3_1}\nonumber
{\bf X}_{k} = \sum^{L}_{l=1} {\bf U}_{k,l} {\bf S}_{k,l},
\end{equation}
where $k$ denotes the user index; $l$ denotes the unit index;
${\bf U}_{k,l} =[{\bf u}^{(k,1)}_{l},{\bf u}^{(k,2)}_{l},\cdots,{\bf u}^{(k,k-1)}_{l},{\bf u}^{(k,k+1)}_{l}, \cdots,$ ${\bf u}^{(k,K)}_{l}] \in \mathbb{C}^{M \times (K-1)}$ denotes the beamforming matrix applied at user $k$ for the $l$-th unit; ${\bf S}_{k,l} = [{\bf s}^{(k,1)T}_{l},{\bf s}^{(k,2)T}_{l},\cdots,{\bf s}^{(k,k-1)T}_{l},{\bf s}^{(k,k+1)T}_{l},\cdots, {\bf s}^{(k,K)T}_{l}]^T \in \mathbb{C}^{(K-1) \times T}$ denotes the transmit spatial streams over $T$ channel uses;
${\bf u}^{(k,k^\prime)}_{l}$ corresponds to the beamformer of spatial stream ${\bf s}^{(k,k^\prime)}_{l}$.
Note that the maximum number of spatial streams in a unit is $K(K-1)$.
But this number can be reduced to $K^\prime(K^\prime-1)$, where $K^\prime(\leq K)$ is the number of active users in the unit.

During the uplink, the equivalent channel matrix from user $k$ to the relay can be expressed by
\begin{equation}\label{Sys_model_3_2}
{\bf H}_{k}{\bf U}_{k,l} = \left[{\bf h}^{(k,1)}_{l}, {\bf h}^{(k,2)}_{l}, \cdots,{\bf h}^{(k,k-1)}_{l},{\bf h}^{(k,k+1)}_{l},\cdots,{\bf h}^{(k,K)}_{l}\right].
\end{equation}
Note that the equivalent channel vectors of unit $l$, i.e., $\{{\bf h}^{(k,k^\prime)}_l| \forall k, k^\prime, k\neq k^\prime \}$, form a spatial structure, referred to as a pattern.

The transmit signal at the relay node can be written as
\begin{equation}\label{Sys_model_3_3}
{\bf X}_R  = {\bf F} {\bf Y}_R,
\end{equation}
where ${\bf F}$ denotes the linear beamforming matrix used at the relay. Similar to the uplink, by using linear receive matrix ${\bf V}_{k,l} =[{\bf v}^{(k,1)}_l,{\bf v}^{(k,2)}_l,\cdots,{\bf v}^{(k,k-1)}_l,{\bf v}^{(k,k+1)}_l, \cdots,{\bf v}^{(k,K)}_l]^T \in \mathbb{C}^{(K-1) \times M}$, the equivalent channel matrix in the downlink is given by
\begin{equation}\label{Sys_model_3_4}
{\bf V}_{k,l}{\bf G}_{k} = \left[{\bf g}^{(k,1)}_l, {\bf g}^{(k,2)}_l, \cdots,{\bf g}^{(k,k-1)}_l,{\bf g}^{(k,k+1)}_l,\cdots,{\bf g}^{(k,K)}_l\right]^T.
\end{equation}
Later, we will show that due to symmetry between the uplink and the downlink, the uplink design straightforwardly carries over to the downlink. Thus, we mostly focus on the uplink design in this paper.

In what follows, we will see that by dividing the spatial streams into a number of units, the signal alignment can be realized in a unit-by-unit fashion, which facilitates the system design.

\subsection{Degrees of Freedom}
Let $R^{(k,k^\prime)}$ be the information rate carried in $W^{(k,k^\prime)}$, and $\hat{W}^{(k,k^\prime)}$ be the estimate of $W^{(k,k^\prime)}$ at user $k^\prime$. We say that user $k$ achieves a sum rate of $C_{k} = \sum_{k^\prime = 1, k^\prime\neq k}^{K} R^{(k,k^\prime)}$, if $\mathrm{Pr}\{\hat{W}^{(k,k^\prime)}\neq W^{(k,k^\prime)}\}$ tends to zero as $T \rightarrow \infty$.

We assume a symmetric mRC with $P_1=P_2=\cdots=P_K=P_R=P$ and $\sigma^2_1 = \sigma^2_2 = \cdots = \sigma^2_K = \sigma^2_R = \sigma^2 $.
Denote ${\rm SNR} =\frac{P}{\sigma^2}$. Let $C_{k}({\rm SNR})$, $k \in \mathcal{I}_K$, be an achievable rate of user $k$. The total DoF of the mRC is defined as
\begin{eqnarray}\nonumber
d_{\mathrm{sum}} \triangleq \lim_{{\rm SNR} \rightarrow \infty} \frac{\sum_{k=1}^{K}C_{k}({\rm SNR})}{\log{{\rm SNR}}}.
\end{eqnarray}
Also, we define the DoF per user and the DoF per relay dimension respectively as
\begin{eqnarray}\label{d_relay}
d_{\mathrm{user}} \triangleq \frac{1}{K} d_{\mathrm{sum}}~~{\rm and}~~d_{\mathrm{relay}} \triangleq \frac{1}{N} d_{\mathrm{sum}}.
\end{eqnarray}

\section{A DoF Outer Bound}\label{Outer_Bound}
An outer bound on the total DoF of the MIMO mRC is given in \cite{Tian_IT2013} as
\begin{subequations}\label{bound_1}
\begin{eqnarray}
d_{\mathrm{sum}} \leq  \min (KM,2N),
\end{eqnarray}
or equivalently
\begin{eqnarray}\label{bound_2}
d_{\mathrm{user}} \leq  \min \left(M,\frac{2N}{K} \right).
\end{eqnarray}
\end{subequations}
The above outer bound can be intuitively explained as follows. On one hand, the total number of independent spatial data streams supported by the MIMO mRC cannot exceed $2N$, as the relay's signal space has $N$ dimensions and thus the relay can only decode and forward $N$ network-coded messages.
On the other hand, the number of independent spatial data streams transmitted or received by each user cannot exceed $M$, as each user only has $M$ antennas. The outer bound in \eqref{bound_1} will be used as a benchmark in the following system design.

\section{MIMO mRC with $K=3$}

In this section, we focus the DoF of the MIMO mRC with $K = 3$. We propose a signal alignment scheme to achieve the DoF capacity of the MIMO mRC with $K=3$ for an arbitrary antenna setup of $(M,N)$.

\subsection{Preliminaries}\label{OneCluster_Preliminiary}

We give some intuitions of the signal alignment by considering only one unit.
For brevity, we omit the unit index $l$ in this subsection.
Recall that $\mathbf{s}^{(k,k^\prime)}$ and $\mathbf{s}^{(k^\prime, k)}$ are exchanged in a pairwise manner for any $k\neq k^\prime$. For convenience, we refer to $\mathbf{s}^{(k,k^\prime)}$ and $\mathbf{s}^{(k^\prime, k)}$ as the signal pair ($k,k^\prime$).
Denote by $\mathbf{h}^{(k,k^\prime)} $ and $\mathbf{g}^{(k,k^\prime)}$ the equivalent channels in the uplink and the downlink, respectively. The system model in \eqref{System-1} and \eqref{System-2} reduces to
\begin{subequations} \label{System11}
\begin{eqnarray}
\mathbf{Y}_R &=& \sum_{k=1}^{K}\sum_{k^\prime=1, k^\prime\neq k}^{K}\mathbf{h}^{(k,k^\prime)} {\mathbf{s}^{(k,k^\prime)T }}+\mathbf{Z}_R \label{System11-a}\\
{\mathbf{y}^{(k, k^\prime)T}_{k}} &=& {\mathbf{g}^{(k,k^\prime)T}} \mathbf{X}_{R}+\mathbf{z}^{(k, k^\prime)T}_{k}, ~~ k \in \mathcal{I}_K. \label{System11-b}
\end{eqnarray}
\end{subequations}
where $\mathbf{z}^{(k,k^\prime)T}_k=\mathbf{v}^{(k,k^\prime)T}\mathbf{Z}_{k}$.
The principle of PNC is applied in relay decoding. Specifically, for each user pair $(k,k^\prime)$, the relay decodes a linear mixture of $\mathbf{s}^{(k,k^\prime)}$ and $\mathbf{s}^{(k^\prime,k)}$ as follows.
Denote by $\mathbf{H}^{(k,k^\prime)} \in \mathbb{C}^{N\times 4}$ the matrix formed by all the uplink channel vectors
except $\mathbf{h}^{(k,k^\prime)}$ and $\mathbf{h}^{(k^\prime, k)}$. Then, define the projection matrix of pair $(k,k^\prime)$ as $\mathbf{P}^{(k,k^\prime)} = \mathbf{I}_N - \mathbf{H}^{(k,k^\prime)} ({\mathbf{H}^{(k,k^\prime)H}} \mathbf{H}^{(k,k^\prime)})^{-1} {\mathbf{H}^{(k,k^\prime)H}} \in \mathbb{C}^{N\times N}$.
For each pair $(k,k^\prime)$, the relay projects the received signal vector ${\mathbf Y}_R$ onto the nullspace of ${\rm span}(\mathbf{H}^{(k,k^\prime)})$, yielding
\begin{equation}\label{PjYR1}
\begin{split}
&\mathbf{P}^{(k,k^\prime)}\mathbf{Y}_R
  =   \mathbf{P}^{(k,k^\prime)}\bigg(\sum_{k=1}^{K}\sum_{k^\prime=1, k^\prime\neq k}^{K}\mathbf{h}^{(k,k^\prime)} {\mathbf{s}^{(k, k^\prime)T}}+\mathbf{Z}_R\bigg) \\ 
 & = \mathbf{P}^{(k,k^\prime)}\big(\mathbf{h}^{(k ,k^\prime)} {\mathbf{s}^{(k,k^\prime)T}} + \mathbf{h}^{(k^\prime, k)} {\mathbf{s}^{(k^\prime, k)T} }\big) + \mathbf{P}^{(k,k^\prime)}\mathbf{Z}_R. 
\end{split}
\end{equation}

We now move to the relay-to-user phase modeled in \eqref{System11-b}. Similarly to $\mathbf{H}^{(k,k^\prime)}$, we denote $\mathbf{G}^{(k,k^\prime)} \in \mathbb{C}^{N\times 4}$ as the matrix formed by all the downlink channel vectors except $\mathbf{g}^{(k,k^\prime)}$ and $\mathbf{g}^{(k^\prime, k)}$. The projection matrix of pair $(k,k^\prime)$ in the downlink is then defined as $\mathbf{W}^{(k,k^\prime)} = \mathbf{I}_N - \mathbf{G}^{(k,k^\prime)}({\mathbf{G}^{(k,k^\prime)H} } \mathbf{G}^{(k,k^\prime)})^{-1} {\mathbf{G}^{(k,k^\prime)H}} \in \mathbb{C}^{N\times N}$. The relay sends out ${\bf F} \mathbf{Y}_R$ with ${\bf F}$ defined in \eqref{Sys_model_3_3} given by
\begin{equation}\label{RelayPrecoder_L1}
{\bf F} = \alpha \sum^{K}_{k=1} \sum^{K}_{k^\prime = k+1} \mathbf{W}^{(k,k^\prime)}\mathbf{P}^{(k,k^\prime)},
\end{equation}
where $\alpha$ is a scaling factor to meet the relay's power constraint.
In \eqref{RelayPrecoder_L1}, the index $k^\prime$ starts from $k+1$ since a project matrix $\mathbf{P}^{(k,k^\prime)}$ is used to extract the signals $\mathbf{s}^{(k,k^\prime)}$ and $\mathbf{s}^{(k^\prime,k)}$ simultaneously.
The received signal at user $k$ is given by
\begin{equation}\label{Yjk1}
\begin{split}
{\mathbf{y}^{(k,k^\prime)T}_k} &= {\mathbf{g}^{(k, k^\prime)T}} \sum_{l=1}^{K} \sum_{n=l+1}^{K} \mathbf{W}^{(l n)}\mathbf{P}^{(ln)}\mathbf{Y}_R + \mathbf{z}^{(k,k^\prime)T}_{k}  \\
&= {\mathbf{g}^{(k,k^\prime)T}} \mathbf{W}^{(k,k^\prime)} \mathbf{P}^{(k,k^\prime)} \left(\mathbf{h}^{(k, k^\prime)} {\mathbf{s}^{(k,k^\prime)T}}  + \right. \\
& ~~~ \left. \mathbf{h}^{(k^\prime, k)} {\mathbf{s}^{(k^\prime, k)T}}  + \mathbf{Z}_R \right)+ \mathbf{z}^{(k,k^\prime)T}_{k}.
\end{split}
\end{equation}
We note that $\mathbf{g}^{(k, k^\prime)}$, $\mathbf{W}^{(k,k^\prime)}$, $\mathbf{P}^{(k,k^\prime)}$, and $\mathbf{h}^{(k^\prime,k)}$ are independent of each other. Therefore, the equivalent user-to-user channel coefficient ${\mathbf{g}^{(k,k^\prime)T}} \mathbf{W}^{(k,k^\prime)}\mathbf{P}^{(k,k^\prime)} \mathbf{h}^{(k^\prime,k)}$ is non-zero with probability one, provided that $\mathbf{W}^{(k,k^\prime)}$ and $\mathbf{P}^{(k,k^\prime)}$ are of at least rank one. Then, each user $k$ receives one linear combination of the two signals in pair ($k,k^\prime$). By subtracting the self-interference, each user can decode the desired messages from the other two users, which achieves a per-user DoF of $d_{\rm user}= 2$, or equivalently, a total DoF of $d_{\mathrm{sum}} = 6$ can be achieved.
From \eqref{Yjk1}, we see that the symmetry exists between the design of the uplink and the design of the downlink. Given the design of the beamformer ${\bf u}^{(k,k^\prime)}$ and the projection matrix $\mathbf{P}^{(k,k^\prime)}$, the receive vector ${\bf v}^{(k,k^\prime)}$ and the projection matrix $\mathbf{W}^{(k,k^\prime)}$ in the downlink can be designed similarly, since ${\mathbf{g}^{(k,k^\prime)T}} \mathbf{W}^{(k,k^\prime)}$ can be simply regarded as the transpose of $\mathbf{P}^{(k,k^\prime)} \mathbf{h}^{(k^\prime,k)}$. Therefore, it suffices to focus on design of the uplink in what follows.

\begin{figure}[tp]
\begin{centering}
\includegraphics[scale=0.5]{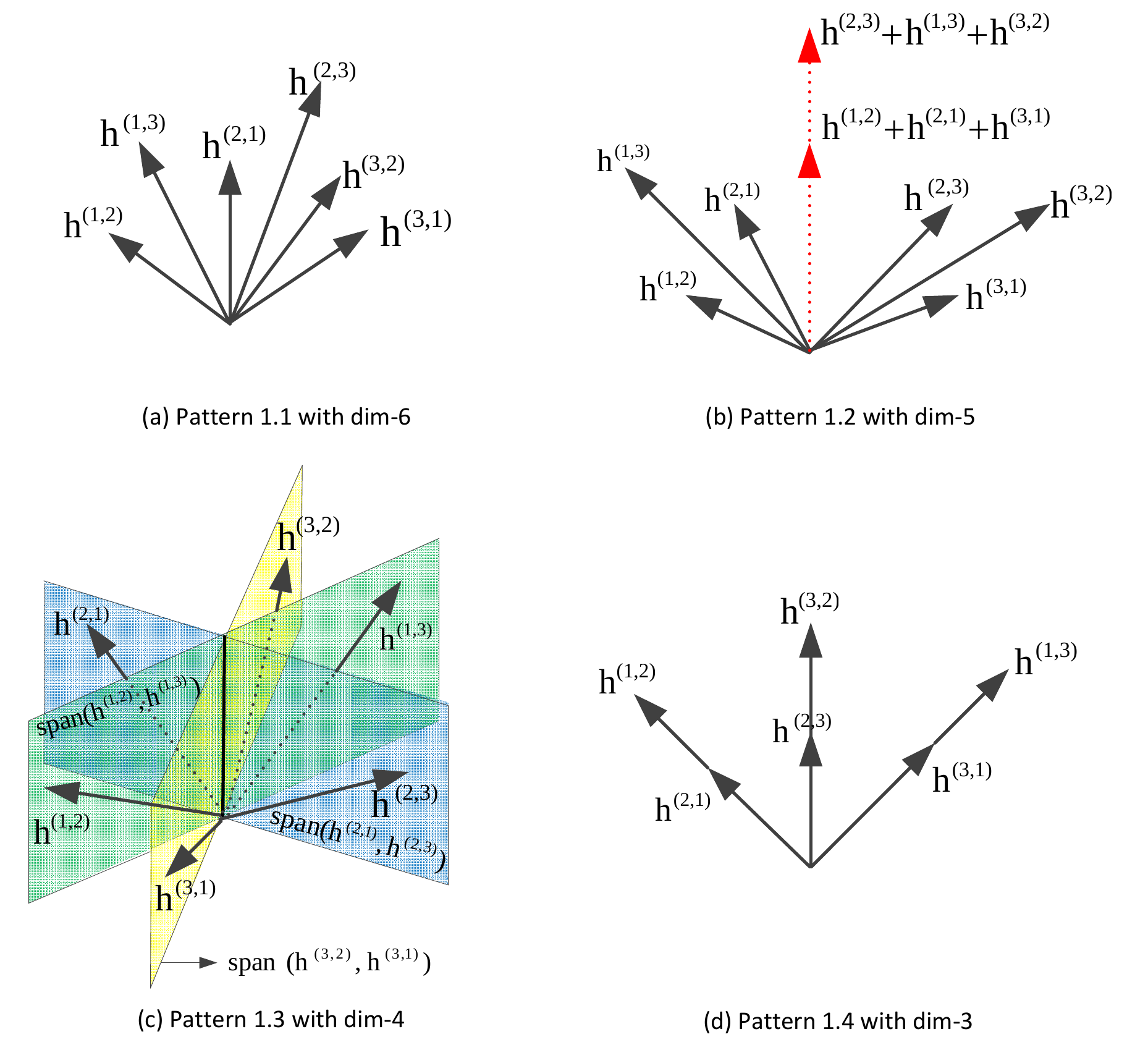}
\vspace{-0.1cm}
\caption{A geometric illustration of Patterns $1.1$ to $1.4$.}  \label{Fig_Pattern_L_1}
\end{centering}
\vspace{-0.3cm}
\end{figure}

We now describe four patterns involved (with a different $d_{\rm relay}$) in achieving the DoF capacity of the MIMO mRC with $K=3$.
Denote $\mathcal{U} \triangleq \{\mathbf{h}^{(k,k^\prime)}|  k \in \mathcal{I}_K, k^\prime \in \mathcal{I}_K, k^\prime \neq k\}$.
Let $\mathcal{U}\backslash\{\mathbf{h}^{(k,k^\prime)},\mathbf{h}^{(k^\prime, k)}\}$ be the vector set obtained by excluding $\mathbf{h}^{(k,k^\prime)}$ and $\mathbf{h}^{(k^\prime, k)}$ from $\mathcal{U}$.

\begin{enumerate}
  \item \textbf{Pattern 1.1:} $\mathcal{U}$ spans a subspace with dimension $6$ (dim-$6$) in $\mathbb{C}^{N}$.
  \item \textbf{Pattern 1.2:} $\mathcal{U}$ spans a subspace with dim-$5$ in $\mathbb{C}^{N}$; for any pair $(k,k^\prime)$, $\mathcal{U}\backslash\{\mathbf{h}^{(k,k^\prime)},\mathbf{h}^{(k^\prime, k)}\}$ spans a subspace of dim-$4$.
   \item \textbf{Pattern 1.3:} $\mathcal{U}$ spans a subspace with dim-$4$ in $\mathbb{C}^{N}$; the intersection of $\mathrm{span}(\mathbf{h}^{(1,2)}, \mathbf{h}^{(2,1)})$, $\mathrm{span}(\mathbf{h}^{(2,3)}, \mathbf{h}^{(3,2)})$, and $\mathrm{span}(\mathbf{h}^{(1,3)}, \mathbf{h}^{(3,1)})$ is of dim-$1$, i.e., these three planes go through a common line, so that $\mathcal{U}$ spans a subspace of dim-$4$.
   \item \textbf{Pattern 1.4:} $\mathcal{U}$ spans a subspace with dim-$3$ in $\mathbb{C}^{N}$; for any pair $(k,k^\prime)$, $\mathbf{h}^{(k,k^\prime)}$ and $\mathbf{h}^{(k^\prime, k)}$ span a subspace of dim-$1$.
\end{enumerate}

The above four patterns are geometrically illustrated in Fig.~\ref{Fig_Pattern_L_1}.
It can be verified that the projection matrices $\mathbf{P}^{(k,k^\prime)}$, $\forall k, k^\prime, k^\prime \neq k$, for Patterns $1.1$-$1.4$ are of at least rank one for sure. For example, $\mathbf{P}^{(k,k^\prime)}$ for Pattern $1.1$ is of at least rank two for sure. Hence the proposed signaling scheme achieves a total DoF of $6$. However,
a different pattern spans a subspace with a different number of dimensions,
which yields a different $d_{\rm relay}$ as shown in Table~\ref{Table OneCluster}. In general, a pattern with a higher $d_{\rm relay}$ is more efficient in utilizing the relay's signal space, and hence is more desirable in the signal alignment design. The requirement on $\frac{M}{N}$ to realize each specific pattern is given in the last column of Table~\ref{Table OneCluster}. Note that these requirements will be discussed in detail in Subsection~\ref{Proof_OneCluster}.
It is also worth noting that Pattern $1.2$ and Pattern $1.3$ have the same requirement on $\frac{M}{N}$, but Pattern $1.3$ achieves a higher $d_{\rm relay}$ than Pattern $1.2$. Thus, Pattern $1.2$ is ruled out by Pattern $1.3$ in the proposed signal alignment scheme.

\begin{table}[!t]
\centering
\caption{Patterns for the MIMO mRC with $K=3$}
\label{Table OneCluster}
\begin{IEEEeqnarraybox}[\IEEEeqnarraystrutmode\IEEEeqnarraystrutsizeadd{2pt}{1pt}]{v/c/v/c/v/c/v/c/v/c/v}
\IEEEeqnarrayrulerow\\
&\mbox{Pattern}&&\mbox{Dimension}&& d_{\mathrm{sum}} && d_{\mathrm{relay}} && \mbox{Requirement} &\\
\IEEEeqnarraydblrulerow\\
\IEEEeqnarrayseprow[3pt]\\
& 1.1 && 6 && 6 && 1 && \frac{M}{N}>0 &\\
\IEEEeqnarrayseprow[3pt]\\
\IEEEeqnarrayrulerow\\
\IEEEeqnarrayseprow[3pt]\\
& 1.2 && 5 && 6 && \frac{6}{5} && \frac{M}{N}>\frac{1}{3}&\\
\IEEEeqnarrayseprow[3pt]\\
\IEEEeqnarrayrulerow\\
\IEEEeqnarrayseprow[3pt]\\
& 1.3 && 4 && 6 && \frac{3}{2} && \frac{M}{N}>\frac{1}{3}&\\
\IEEEeqnarrayseprow[3pt]\\
\IEEEeqnarrayrulerow\\
\IEEEeqnarrayseprow[3pt]\\
& 1.4 && 3 && 6 && 2 && \frac{M}{N}>\frac{1}{2}&\\
\IEEEeqnarrayseprow[3pt]\\
\IEEEeqnarrayrulerow
\end{IEEEeqnarraybox}
\end{table}

\begin{figure}[tp]
\begin{centering}
\includegraphics[scale=0.5]{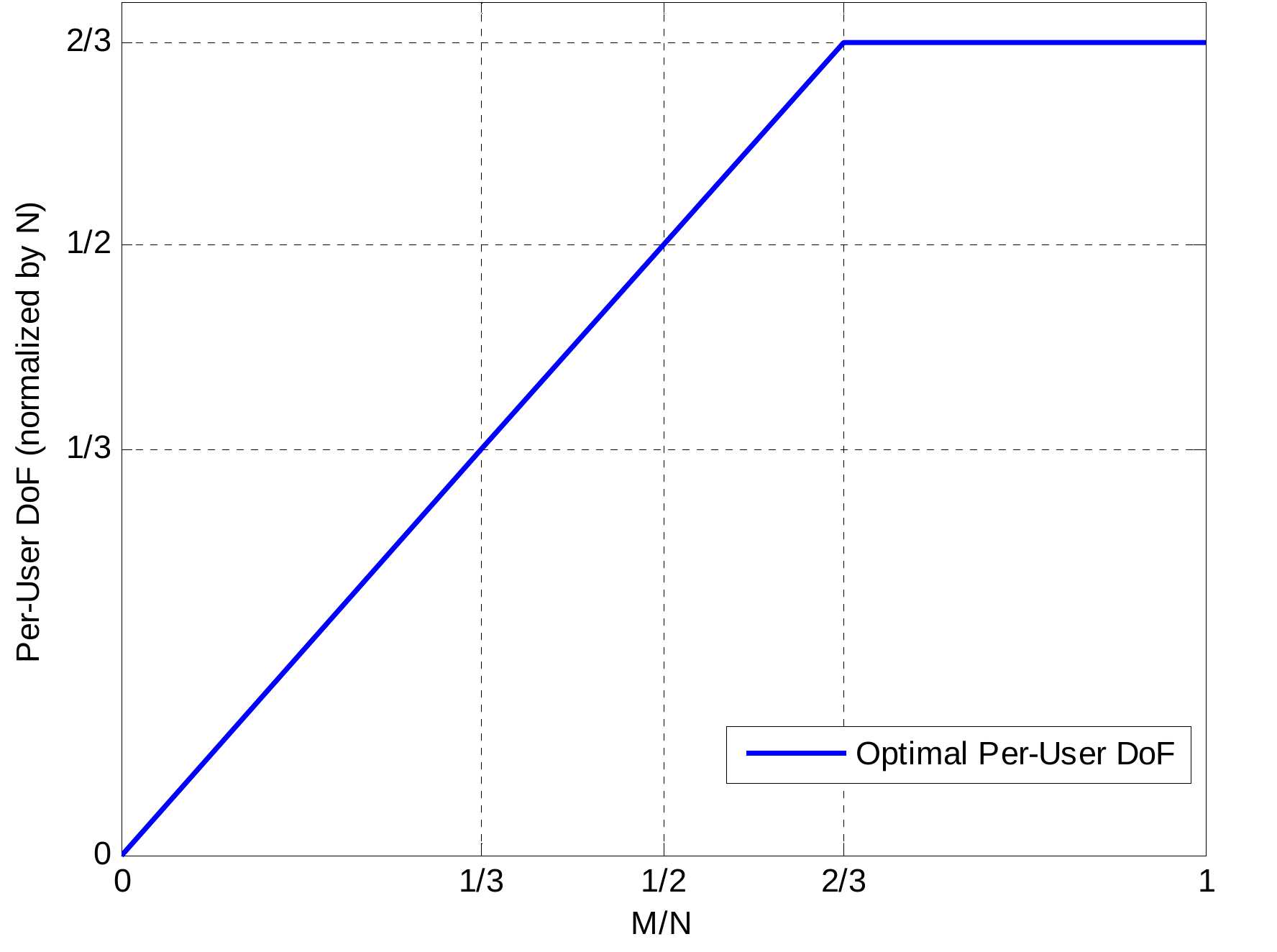}
\vspace{-0.1cm}
\caption{The DoF capacity for the MIMO mRC with $K=3$ against the antenna ratio $\frac{M}{N}$.}  \label{Fig_pairwise_L_1}
\end{centering}
\vspace{-0.3cm}
\end{figure}

\subsection{Main Result}
We now consider the general case that each user transmits multiple spatial data streams over a MIMO mRC, i.e., multiple units co-exist in the relay's signal space with each unit consisting of $K(K-1)$ spatial streams. Our goal is to construct units with the most efficient patterns to occupy the relay's signal space.
The main result is summarized in the following theorem.

\begin{theorem}\label{Theorem OneCluster}
For the MIMO mRC $(M,N,K)$ with $K = 3$, the DoF capacity per user is given by
\begin{equation}\label{Theorem_K_3}
d_{\mathrm{user}} =
\left\{
\begin{aligned}
& M, &\frac{M}{N} < \frac{2}{3}\\
&\frac{2N}{3}, &\frac{M}{N} \geq \frac{2}{3}.
\end{aligned}
\right.
\end{equation}
\end{theorem}

The per-user DoF capacity with respect to $\frac{M}{N}$ is shown in Fig.~\ref{Fig_pairwise_L_1}. We see that the per-user DoF of $d_{\rm user} = M$ is achieved for $\frac{M}{N} < \frac{2}{3}$, which means that the DoF is bounded by the number of antennas at the user ends.
On the other hand, when $\frac{M}{N} \geq \frac{2}{3}$, the DoF is bounded by the number of relay antennas.
Note that the DoF capacity of the three-user MIMO Y channel has been previously derived in \cite{Chaaban_ISIT2013}. However, we emphasize that the proposed signal alignment technique in our proof (cf., equations \eqref{OneCluster-9-0}-\eqref{Appdix-3-0} and the discussions therein) is very different from the one in \cite{Chaaban_ISIT2013}. Also, our proposed technique can be extended for the case of an arbitrary $K$, which is the major contribution of this paper.

\subsection{Proof of Theorem \ref{Theorem OneCluster}}\label{Proof_OneCluster}
We first note that $d_{\rm user}$ in \eqref{Theorem_K_3} coincides with the DoF outer bound in \eqref{bound_1} with $K=3$.
Thus, to prove Theorem~\ref{Theorem OneCluster}, it suffices to show the achievability of \eqref{Theorem_K_3}.
We start with a brief description of the overall transceiver design.
We need to jointly design the transmit beamformers $\{ {\bf u}^{(k,k^\prime)}_l\}$, the receive vectors $\{{\bf v}^{(k,k^\prime)}_l\}$, and the relay projection matrices $\{\mathbf{P}^{(k,k^\prime)}_l \}$ and $\{\mathbf{W}^{(k,k^\prime)}_l \}$ to efficiently utilize the relay's signal space.
As different from \eqref{RelayPrecoder_L1}, here
the relay's projection matrices $\mathbf{P}^{(k,k^\prime)}_l$ and $\mathbf{W}^{(k,k^\prime)}_l$ null the interference not only from the other pairs in unit $l$ but also from the other units.
Taking $\mathbf{P}^{(k,k^\prime)}_l$ as an example, we see that it projects a vector into the null space of ${\rm span}\big(\{ {\bf H}_{\bar{k}}{\bf u}^{(\bar{k},\bar{k}^\prime)}_{\bar{l}}|\forall \bar{l},\bar{k},\bar{k}^\prime; \bar{k}^\prime \neq \bar{k} \} \backslash \{ {\bf H}_{k}{\bf u}^{(k,k^\prime)}_l, {\bf H}_{k}{\bf u}^{(k^\prime,k)}_l\}\big)$.
Hence the relay beamforming matrix $\bf F$ given in \eqref{Sys_model_3_3} is expressed as
\begin{equation}\label{RelayPrecoder_L1-0}
{\bf F} = \alpha \sum^{L}_{l=1} \sum^{K}_{k=1} \sum^{K}_{k^\prime = k+1} \mathbf{W}^{(k,k^\prime)}_l \mathbf{P}^{(k,k^\prime)}_l,
\end{equation}
where $L$ denotes the number of the units. Based on that, in each unit, each user achieves a DoF of two, provided that the projection matrices $\mathbf{P}^{(k,k^\prime)}_l$ and $\mathbf{W}^{(k,k^\prime)}_l$ are at least of rank one.
Note that with the transmit beamformer ${\bf u}^{(k,k^\prime)}_l$ and the receive vector ${\bf v}^{(k,k^\prime)}_l$,
the equivalent channel regarding the spatial stream $\mathbf{s}^{(k,k^\prime)}_l$ in the uplink is denoted as $\mathbf{h}^{(k,k^\prime)}_{l}={\bf H}_{k}\mathbf{u}^{(k,k^\prime)}_l$, and the equivalent channel vector regarding the spatial stream $\mathbf{s}^{(k^\prime,k)}_l$ in the downlink is denoted by $\mathbf{g}^{(k,k^\prime)}_l={\bf G}^T_{k}\mathbf{v}^{(k,k^\prime)}_l$.
Next we derive the DoF given in Theorem~\ref{Theorem OneCluster} by dividing the overall range of $\frac{M}{N}$ into multiple intervals.

\subsubsection{Case of $\frac{M}{N} \leq \frac{1}{3}$}
In this case, the number of relay antennas $N$ is no less than the number of antennas of all the users, i.e., ${\rm span}\left(\{{\bf H}_k|\forall k \} \right)$ is of dim-$3M$ with probability one. This implies that the relay's signal space has enough dimensions to support full multiplexing at the user end, i.e., each user transmits $M$ spatial streams.
$M$ units with Pattern $1.1$ (as shown in Fig.~\ref{Fig_Pattern_L_1}(a)) can be constructed. The geometric structure in Fig.~\ref{Fig_Pattern_L_1}(a) indicates that each spatial stream in a unit occupies an independent direction in the relay's signal space. Then, in total the $6$ spatial streams in a unit span a subspace of dim-$6$.
As the directions of signals with Pattern $1.1$ are randomly drawn from the relay's signal space, the independence of different units can be guaranteed with probability one. Clearly, the projection matrix $\mathbf{P}^{(k,k^\prime)}_l$ is of at least rank one. Thus, each unit achieves a DoF of $6$. Considering all the $M$ units, we obtain that the achievable per-user DoF is $M$.\footnote{{A similar proof of $d_{\rm user}=M$ for $\frac{M}{N}\leq \frac{1}{3}$ can be found in \cite{Tian_IT2013}.}}

\subsubsection{Case of $\frac{1}{3} < \frac{M}{N} \leq \frac{1}{2}$}
From Table~\ref{Table OneCluster}, this case corresponds to Patterns $1.2$ and $1.3$.
As Pattern $1.3$ is more efficient than Pattern $1.2$\footnote{Pattern $1.2$ can be constructed by designing the transmit beamforming vectors in a unit as
\begin{eqnarray}\label{OneCluster-7}\nonumber
 {\bf H}_1 ({\bf u}^{(1,2)}_l + {\bf u}^{(1,3)}_l ) + {\bf H}_2 ({\bf u}^{(2,1)}_l + {\bf u}^{(2,3)}_l)
  + {\bf H}_3 ({\bf u}^{(3,1)}_l + {\bf u}^{(3,2)}_l)  = {\bf 0},
\end{eqnarray}
or equivalently
\begin{eqnarray}\label{OneCluster-7}\nonumber
 \left({\bf h}^{(1,2)}_l +{\bf h}^{(2,1)}_l +{\bf h}^{(3,1)}_l \right) +  \left({\bf h}^{(2,3)}_l +{\bf h}^{(1,3)}_l +{\bf h}^{(3,2)}_l \right) = {\bf 0},
\end{eqnarray}
which implies that the direction of ${\bf h}^{(1,2)}_l+{\bf h}^{(2,1)}_l+{\bf h}^{(3,1)}_l$ is parallel to the direction of ${\bf h}^{(2,3)}_l+{\bf h}^{(1,3)}_l+{\bf h}^{(3,2)}_l$ following Pattern $1.2$ in Fig.~\ref{Fig_Pattern_L_1}(b). Thus, the $6$ spatial streams in a unit span a subspace of dim-$5$.}
(i.e., the former achieves a higher $d_{\rm relay}$ than the latter), we focus on the construction of units following Pattern $1.3$.
Denote the intersection of $\mathrm{span}(\mathbf{H}_{1}, \mathbf{H}_{2})$ and $\mathrm{span}(\mathbf{H}_{3})$ by $\mathcal{S}^{(1.3)}$. The dimension of $\mathcal{S}^{(1.3)}$ is $3M-N > 0$.
We choose two vectors $\mathbf{u}_{3,l}$ and $\mathbf{u}^{(3,1)}_l$ such that $\mathbf{H}_{3}\mathbf{u}_{3,l}$ and $\mathbf{H}_{3}\mathbf{u}^{(3,1)}_l$ are two linearly independent vectors in $\mathcal{S}^{(1.3)}$. By definition, both $\mathbf{H}_{3}\mathbf{u}_{3,l}$ and $\mathbf{H}_{3}\mathbf{u}^{(3,1)}_l$ belong to $\mathrm{span}(\mathbf{H}_{1}, \mathbf{H}_{2}) = {\rm span}({\bf H}_1) \oplus {\rm span}({\bf H}_2)$. Thus, there uniquely exist $\{\mathbf{u}^{(1,3)}_l, \mathbf{u}^{(2,3)}_l\}$ and $\{\mathbf{u}_{1,l}, \mathbf{u}^{(2,1)}_l\}$ satisfying
\begin{subequations}\label{OneCluster-9-0}
\begin{eqnarray}
\mathbf{H}_{1}\mathbf{u}^{(1,3)}_l+\mathbf{H}_{2}\mathbf{u}^{(2,3)}_l+\mathbf{H}_{3}\mathbf{u}_{3,l} = \mathbf{0} \label{OneCluster-9-0-a}\\
\mathbf{H}_{1}\mathbf{u}_{1,l}+\mathbf{H}_{2}\mathbf{u}^{(2,1)}_l+\mathbf{H}_{3}\mathbf{u}^{(3,1)}_l = \mathbf{0}. \label{OneCluster-9-0-b}
\end{eqnarray}
\end{subequations}
Let $\mathbf{u}^{(3,2)}_l = \mathbf{u}_{3,l} - \mathbf{u}^{(3,1)}_l$ and $\mathbf{u}^{(1,2)}_l = \mathbf{u}_{1,l} - \mathbf{u}^{(1,3)}_l$. Together with \eqref{OneCluster-9-0}, we obtain
\begin{subequations}\label{OneCluster-9}
\begin{eqnarray}
\mathbf{H}_{1}\mathbf{u}^{(1,3)}_l+\mathbf{H}_{2}\mathbf{u}^{(2,3)}_l+\mathbf{H}_{3}(\mathbf{u}^{(3,2)}_l+\mathbf{u}^{(3,1)}_l) = \mathbf{0}\label{OneCluster-9-a}\\
\mathbf{H}_{1}(\mathbf{u}^{(1,2)}_l+\mathbf{u}^{(1,3)}_l)+\mathbf{H}_{2}\mathbf{u}^{(2,1)}_l+\mathbf{H}_{3}\mathbf{u}^{(3,1)}_l = \mathbf{0}.\label{OneCluster-9-b}
\end{eqnarray}
Subtracting \eqref{OneCluster-9-b} by \eqref{OneCluster-9-a}, we further obtain
\begin{eqnarray}\label{OneCluster-9-c}
\mathbf{H}_{1}\mathbf{u}^{(1,2)}_l+\mathbf{H}_{2}(\mathbf{u}^{(2,1)}_l-\mathbf{u}^{(2,3)}_l)-\mathbf{H}_{3}\mathbf{u}^{(3,2)}_l = \mathbf{0}.
\end{eqnarray}
\end{subequations}
We now show that three signal direction pairs $\{ {\bf H}_1\mathbf{u}^{(1,2)}_l,{\bf H}_2\mathbf{u}^{(2,1)}_l\}$, $\{{\bf H}_2\mathbf{u}^{(2,3)}_l, {\bf H}_3\mathbf{u}^{(3,2)}_l \}$, $\{ {\bf H}_3\mathbf{u}^{(3,1)}_l, {\bf H}_1\mathbf{u}^{(1,3)}_l \}$ form a unit with Pattern $1.3$ as shown in Fig.~\ref{Fig_Pattern_L_1}(c).
From \eqref{OneCluster-9-a}, we see that two signal pairs $(1,3)$ and $(2,3)$ span a subspace of dim-$3$, which implies that the plane spanned by signal pair $(1,3)$ and the plane spanned by signal pair $(2,3)$ go through a common line. Further, from \eqref{OneCluster-9-b} and \eqref{OneCluster-9-c}, we see that the plane spanned by signal pair of $(1,2)$ also goes through this common line, which implies that  $\{\mathbf{H}_{k}\mathbf{u}^{(k,k^\prime)}_l|\forall k, k^\prime \neq k\}$ span a subspace of dim-$4$.
Base on that, the dimension of ${\rm null}\big(\{{\bf H}_{1}\mathbf{u}^{(1,3)}_l, {\bf H}_{3}\mathbf{u}^{(3,1)}_l,{\bf H}_{2}\mathbf{u}^{(2,3)}_l,{\bf H}_{3}\mathbf{u}^{(3,2)}_l \} \big)\cap {\rm span}\big(\{\mathbf{H}_{k}\mathbf{u}^{(k,k^\prime)}_l|\forall k, k^\prime \neq k\}\big)$ is of dim-$1$.
Similarly, from \eqref{OneCluster-9-b} and \eqref{OneCluster-9-c}, the intersection of nullspace of any two of the three pairs in unit $l$ and the subspace spanned by signals in unit $l$ is of dim-$1$. Therefore, a linear combination of the signals for each signal pair can be decoded at the relay without interference.

We now describe how to construct multiple linearly independent units following Pattern $1.3$.
Let the columns of ${\bf U}_H \in \mathbb{C}^{3M\times (3M-N)}$ be a basis of ${\rm null}\left(\left[{\bf H}_1,{\bf H}_2,{\bf H}_3\right]\right)$. Partition ${\bf U}_H$ as ${\bf U}_H =\left[{\bf U}^T_{H_1}, {\bf U}^T_{H_2},{\bf U}^T_{H_3} \right]^T$ with ${\bf U}_{H_i} \in \mathbb{C}^{M\times (3M-N)}$. Then, ${\bf u}^{(1,3)}_l$, ${\bf u}^{(2,3)}_l$, and ${\bf u}_{3,l}$ in \eqref{OneCluster-9-0-a}, are respectively chosen as the $(2l-1)$-th column of ${\bf U}_{H_1}$, ${\bf U}_{H_2}$, and ${\bf U}_{H_3}$. Further, ${\bf u}_{1,l}$, ${\bf u}^{(2,1)}_l$, and ${\bf u}^{(3,1)}_l$ in \eqref{OneCluster-9-0-b} are respectively chosen as the $(2l)$-th column of ${\bf U}_{H_1}$, ${\bf U}_{H_2}$, and ${\bf U}_{H_3}$.
From \eqref{OneCluster-9-0}, we see that
${\rm span}\big({\bf H}_1{\bf u}^{(1,3)}_l, {\bf H}_1{\bf u}^{(2,3)}_l,{\bf H}_3{\bf u}_{3,l}, {\bf H}_1{\bf u}_{1,l}, {\bf H}_2{\bf u}^{(2,1)}_l, {\bf H}_3{\bf u}^{(3,1)}_l \big)$ is of dim-$4$, and
\begin{equation}\label{Appdix-3-0}
\begin{split}
& {\rm span}\big({\bf H}_1{\bf u}^{(1,3)}_l, {\bf H}_1{\bf u}^{(2,3)}_l,{\bf H}_3{\bf u}_{3,l}, {\bf H}_1{\bf u}_{1,l}, \\
& ~~~~~~~ {\bf H}_2{\bf u}^{(2,1)}_l, {\bf H}_3{\bf u}^{(3,1)}_l \big)\\
 =~~& {\rm span}\big({\bf H}_1{\bf u}^{(1,3)}_l,{\bf H}_1{\bf u}^{(1,2)}_l,{\bf H}_2{\bf u}^{(2,1)}_l,{\bf H}_2{\bf u}^{(2,3)}_l,\\
&~~~~~~~ {\bf H}_3{\bf u}^{(3,1)}_l,{\bf H}_3{\bf u}^{(3,2)}_l \big),
\end{split}
\end{equation}
by noting $\mathbf{u}^{(3,2)}_l = \mathbf{u}_{3,l} - \mathbf{u}^{(3,1)}_l$ and $\mathbf{u}^{(1,2)}_l = \mathbf{u}_{1,l} - \mathbf{u}^{(1,3)}_l$. Thus, each unit $l$ spans a subspace of dim-$4$.
From Lemma~\ref{Lemma 4} in Appendix~\ref{Appendix 1}, we see that the dimension of ${\rm span}\left({\bf H}_1{\bf U}_{H_1},{\bf H}_2{\bf U}_{H_2},{\bf H}_3{\bf U}_{H_3}\right)$ is $\min \left(2(3M-N),N\right)$. Therefore, we can construct $\min \left(\frac{3M-N}{2}, \frac{N}{4}\right)$ linearly independent units following Pattern $1.3$.\footnote{Here we assume that $\frac{3M-N}{2}$ is an integer. Otherwise, we use symbol extension \cite{Jafar08} to ensure that the dimension of the above intersection is dividable by two; see Appendix~\ref{Appendix 3} for details. Note that the symbol extension is used to achieved a fractional DoF throughout of the rest of this paper without further explicit notification.}

Suppose $2(3M-N)\leq N$, or equivalently, $\frac{M}{N}\leq \frac{1}{2}$. All $\frac{3M-N}{2}$ units span a subspace of $\frac{3M-N}{2}\times 4 = 2(3M-N)$ dimensions.
The remaining $N-2(3M-N)$ dimensions are used to construct $\frac{N-2(3M-N)}{6}$ units with Pattern $1.1$. Thus, the achievable per-user DoF is given by
\begin{eqnarray}\label{OneCluster-11}
d_{\mathrm{user}} = 3M-N + \frac{2(N-2(3M-N))}{6} = M.
\end{eqnarray}
Recall that the directions of signals with Pattern $1.1$ are randomly drawn from the relay's signal space, the independence of the units with Pattern $1.3$ and the units with Pattern $1.1$ can be guaranteed with probability one.
If the overall relay's signal space is occupied by units with Pattern $1.3$, the maximum achievable per-user DoF is $\frac{N}{4}\times 2 =\frac{N}{2}$. Therefore, the achievable per-user DoF is given by $\min (M, \frac{N}{2}) = M$.

\subsubsection{Case of $ \frac{M}{N} >\frac{1}{2}$}
In this case, $\frac{M}{N}$ is large enough to construct units with Pattern $1.4$. The intersection of ${\rm span}({\bf H}_k)$ and ${\rm span}({\bf H}_{k'})$ is of $2M-N >0$.
Let ${\bf h}^{(k,k^\prime)}_l$ be a vector in the intersection of ${\rm span}({\bf H}_k)$ and ${\rm span}({\bf H}_{k'})$. There exist $\{{\bf u}^{(k,k')}_l, {\bf u}^{(k',k)}_l\}$ satisfying
\begin{eqnarray}\label{OneCluster-12} \nonumber
{\bf H}_k {\bf u}^{(k,k')}_l = {\bf H}_{k'} {\bf u}^{(k',k)}_l = {\bf h}^{(k,k')}_l,~\forall k, k \neq k',
\end{eqnarray}
which implies that the two spatial streams of pair $(k,k^\prime)$ in a unit span a subspace of dim-$1$, i.e., two spatial streams of a pair are aligned in one direction as illustrated in Fig.~\ref{Fig_Pattern_L_1}(d). Then, in total the $6$ spatial streams in a unit occupy a subspace of dim-$3$.
In this way, we can construct $2M-N$ units with Pattern $1.4$,
in total spanning a subspace of dim-$3(2M-N)$.
According to Lemma~\ref{Lemma 3} in Appendix~\ref{Appendix 1}, we obtain that $\{{\bf h}^{(k,k')}_l| \forall l\}$ are linearly independent with probability one. Further, due to the randomness of ${\bf H}_k$, the independence of $\{{\bf h}^{(k,k')}_l| \forall l,k,k^\prime, k\neq k^\prime\}$ are guaranteed with probability one.
Again, the remaining $N-3(2M-N)$ dimensions can be used for constructing $\frac{N-3(2M-N)}{4}$ units with Pattern $1.3$. Thus an achievable per-user DoF is given by
\begin{eqnarray}\label{OneCluster-13} \nonumber
d_{\mathrm{user}} = 2(2M-N) + \frac{2(N-3(2M-N))}{4} = M.
\end{eqnarray}
When the overall relay's signal space is occupied by the units with Pattern $1.4$, a per-user DoF of $\frac{N}{3}\times 2 = \frac{2N}{3}$ is achieved. Therefore, the maximum achievable per-user DoF is given by $\min (M, \frac{2N}{3})$.\footnote{A similar proof of $d_{\rm user}=\frac{2N}{3}$ for $\frac{M}{N}\geq \frac{2}{3}$ can be found in \cite{Lee10,Tian_IT2013}.}

The above obtained DoF coincides with the DoF upper bound in \eqref{bound_2}, and therefore, this achievable DoF is exactly the DoF capacity of the channel, which concludes the proof of Theorem~\ref{Theorem OneCluster}.

\section{MIMO mRC with $K>3$}
In this section, we generalize Theorem~\ref{Theorem OneCluster} to the case of an arbitrary number of users.
We start with the case of $K=4$.

\subsection{Preliminaries}
Again, we start with some intuitions of the signal alignment by assuming that each user transmits one independent spatial stream to each of the other users in a unit. The relay's beamforming matrix is still given by \eqref{RelayPrecoder_L1}.

The following patterns are involved in deriving the achievable DoF to be presented later.
Also we omit the unit index $l$ for brevity in this subsection.
Denote $\mathcal{U} \triangleq \{\mathbf{h}^{(k,k^\prime)}| k \in \mathcal{I}_K, k^\prime \in \mathcal{I}_K; k\neq k^\prime\}$ with $ \mathcal{I}_K=\{1,2,3,4\}$, and $\mathcal{U}\backslash\{\mathbf{h}^{(k,k^\prime)},\mathbf{h}^{(k^\prime, k)}\}$ is the vector set obtained by excluding $\mathbf{h}^{(k,k^\prime)}$ and $\mathbf{h}^{(k^\prime, k)}$ from $\mathcal{U}$.
Let $\mathcal{U}_{\bar{i}} \triangleq \{\mathbf{h}^{(k,k^\prime)}| k \in \mathcal{I}_K, k^\prime \in \mathcal{I}_K; k\neq k^\prime; k \neq i, k^\prime \neq i\}$ and $\mathcal{U}_{\bar{i}}\backslash\{\mathbf{h}^{(k,k^\prime)},\mathbf{h}^{(k^\prime, k)}\}$ be the vector set obtained by excluding $\mathbf{h}^{(k,k^\prime)}$ and $\mathbf{h}^{(k^\prime, k)}$ from $\mathcal{U}_{\bar{i}}$.
\begin{enumerate}
  \item \textbf{Pattern 2.1:} $\mathcal{U}$ spans a subspace with dim-$12$ in $\mathbb{C}^{N}$.
  \item \textbf{Pattern 2.2:} $\mathcal{U}$ spans a subspace with dim-$9$ in $\mathbb{C}^{N}$; for any pair $(k,k^\prime)$, $\mathcal{U}\backslash\{\mathbf{h}^{(k,k^\prime)},\mathbf{h}^{(k^\prime, k)}\}$ spans a subspace with dim-$8$.
  \item \textbf{Pattern 2.3:}
For each $i$, $\mathcal{U}_{\bar{i}}$ spans a subspace of dim-$4$ in $\mathbb{C}^{N}$ following Pattern $1.3$.
  \item \textbf{Pattern 2.4:} $\mathcal{U}$ spans a subspace with dim-$6$ in $\mathbb{C}^{N}$; for any pair $(k,k^\prime)$, $\mathbf{h}^{(k,k^\prime)}$ and $\mathbf{h}^{(k^\prime, k)}$ span a subspace with dim-$1$.
\end{enumerate}
It can be readily shown that the projection matrix $\mathbf{P}^{(k,k^\prime)}$ corresponding to Patterns $2.1$ to $2.4$ are of at least rank one with probability one. Thus, Patterns $2.1$, $2.2$, and $2.4$ achieve a total DoF of $12$, while Pattern $2.3$ achieves a total DoF of $24$. Corresponding antenna requirement for each pattern is given in Table~\ref{Table K4}, which will be discussed in details in the Subsection~\ref{OneCluster-K4}.
It is worth mentioning that for a same requirement on $\frac{M}{N}$, some other patterns may possibly be constructed. However, they are ruled out due to a relatively low $d_{\rm relay}$, i.e., less efficiency in utilizing the relay's signal space.
Again, the downlink patterns are omitted due to the uplink/downlink symmetry.

\begin{table}[!t]
\centering
\caption{Patterns for the MIMO mRC with $K=4$}
\label{Table K4}
\begin{IEEEeqnarraybox}[\IEEEeqnarraystrutmode\IEEEeqnarraystrutsizeadd{2pt}{1pt}]{v/c/v/c/v/c/v/c/v/c/v}
\IEEEeqnarrayrulerow\\
&\mbox{Pattern}&&\mbox{Dimension}&& d_{\mathrm{sum}} && d_{\mathrm{relay}} && \mbox{Requirement} &\\
\IEEEeqnarraydblrulerow\\
\IEEEeqnarrayseprow[3pt]\\
& 2.1 && 12 && 12 && 1 && \frac{M}{N}>0 &\\
\IEEEeqnarrayseprow[3pt]\\
\IEEEeqnarrayrulerow\\
\IEEEeqnarrayseprow[3pt]\\
& 2.2 && 9 && 12 && \frac{4}{3} && \frac{M}{N}>\frac{1}{4}&\\
\IEEEeqnarrayseprow[3pt]\\
\IEEEeqnarrayrulerow\\
\IEEEeqnarrayseprow[3pt]\\
& 2.3 && 4 && 6 && \frac{3}{2} && \frac{M}{N}>\frac{1}{3}&\\
\IEEEeqnarrayseprow[3pt]\\
\IEEEeqnarrayrulerow\\
\IEEEeqnarrayseprow[3pt]\\
& 2.4 && 6 && 12 && 2 && \frac{M}{N}>\frac{1}{2}&\\
\IEEEeqnarrayseprow[3pt]\\
\IEEEeqnarrayrulerow
\end{IEEEeqnarraybox}
\end{table}

\subsection{Main Result}\label{subsection_B_K_4}
\begin{prop}\label{Theorem K4}
For the MIMO mRC $(M,N,K)$ with $K = 4$, the per-user DoF capacity of $d_{\rm user}=M$ is achieved when $ \frac{M}{N} \leq \frac{3}{8}$,
and the per-user DoF capacity of $d_{\rm user} =\frac{N}{2}$ is achieved when $\frac{M}{N} \geq \frac{7}{12}$. For $\frac{M}{N} \in (\frac{3}{8}, \frac{7}{12}]$,
an achievable per-user DoF is given by
\begin{equation}\nonumber
d_{\mathrm{user}} =
\left\{
\begin{aligned}
& \frac{3N}{8}, &\frac{3}{8} < \frac{M}{N} \leq \frac{1}{2}\\
&\frac{3M}{2}-\frac{3N}{8}, &\frac{1}{2} < \frac{M}{N} \leq \frac{7}{12}.
\end{aligned}
\right.
\end{equation}
\end{prop}

\begin{figure}[tp]
\begin{centering}
\includegraphics[scale=0.5]{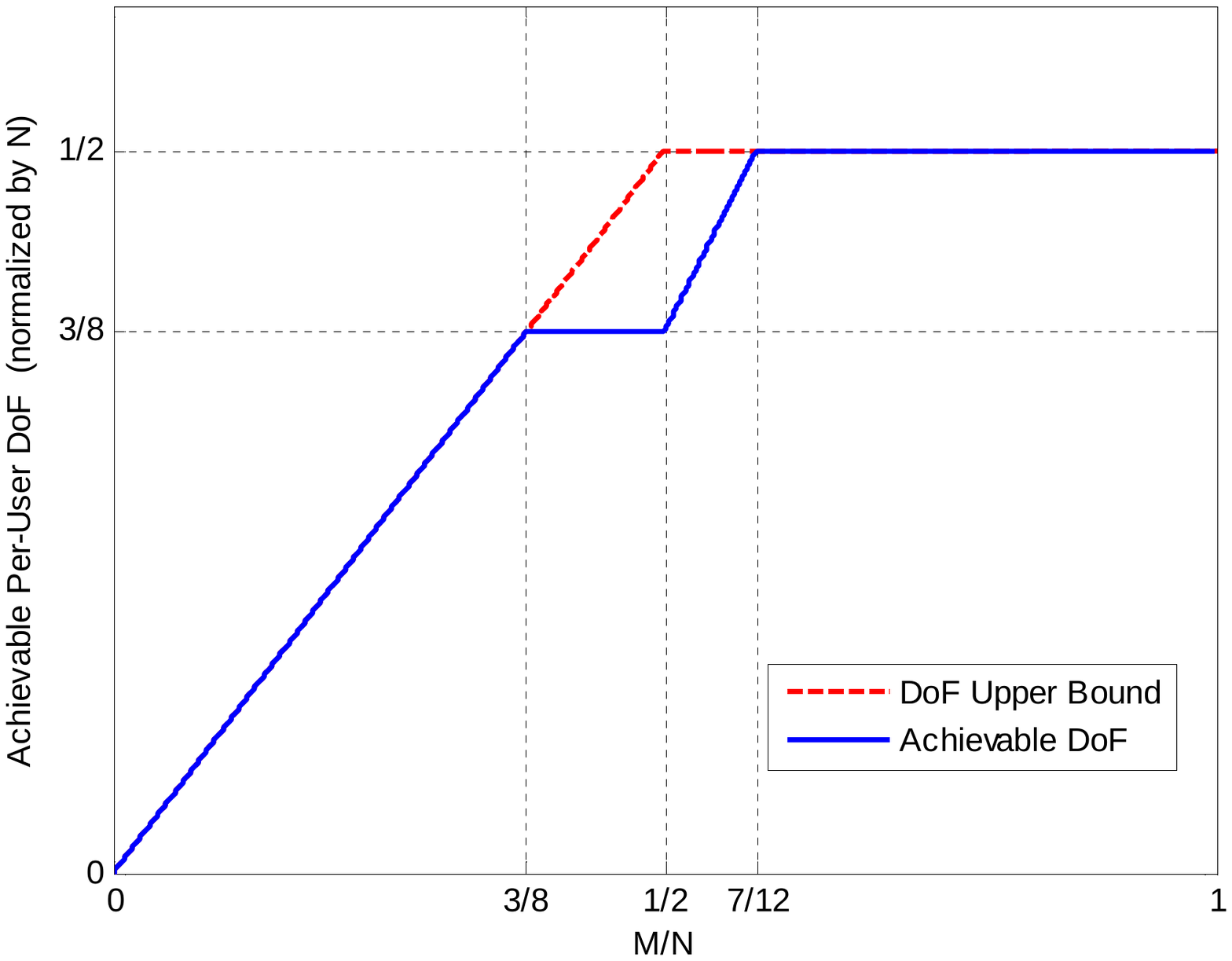}
\vspace{-0.1cm}
\caption{An achievable per-user DoF for the MIMO mRC with $K=4$ against the antenna ratio $\frac{M}{N}$.}  \label{L_1_K_4}
\end{centering}
\vspace{-0.3cm}
\end{figure}

The achievable per-user DoF for MIMO mRC with $K=4$ is illustrated in Fig.~\ref{L_1_K_4}. We observe that, different from the case of $K=3$, the DoF bound given in Section~\ref{Outer_Bound} can only be achieved in the ranges of $\frac{M}{N} \in(0, \frac{3}{8}]$ and $\frac{M}{N}\in [\frac{7}{12}, \infty)$; for $\frac{M}{N} \in (\frac{3}{8}, \frac{7}{12})$, there is a certain DoF gap between the achievable DoF and the capacity outer bound.

\subsection{Proof of Proposition \ref{Theorem K4}}\label{OneCluster-K4}
To prove Proposition~\ref{Theorem K4}, we consider four cases detailed below.
\subsubsection{Case of $\frac{M}{N} \leq \frac{1}{4}$}
In this case, since $N \geq 4M$, the relay's signal space has enough dimensions to support full multiplexing at the users, which implies that each user can transmit $M$ independent spatial streams, or equivalently,
$M$ units with Pattern $2.1$ can be constructed.
Therefore, a per-user DoF of $M$ can be achieved.

\subsubsection{Case of $\frac{1}{4} < \frac{M}{N} \leq \frac{1}{3}$}
As shown in Table~\ref{Table K4}, this case corresponds to Pattern $2.2$. As $4M-N>0$, the nullspace of ${\rm span}\big({\bf H}_1,{\bf H}_2,{\bf H}_3,{\bf H}_4 \big)$ is of dim-$(4M-N)$. Let the columns of ${\bf U}_H \in \mathbb{C}^{N \times (4M-N)}$ be a basis of ${\rm null}\big([{\bf H}_1,{\bf H}_2,{\bf H}_3,{\bf H}_4] \big)$. Partition ${\bf U}_H$ as ${\bf U}_H = [{\bf U}^T_{H_1}, {\bf U}^T_{H_2}, {\bf U}^T_{H_3}, {\bf U}^T_{H_4}]^T$.
From Lemma~\ref{Lemma 4}, ${\rm span}\big({\bf H}_1{\bf U}_{H_1},{\bf H}_2{\bf U}_{H_2},{\bf H}_3{\bf U}_{H_3},{\bf H}_4 {\bf U}_{H_4} \big)$ is of dim-$3(4M-N)$ for sure. Arbitrarily choose three columns of ${\bf U}_H$, denoted by $[\mathbf{u}^{1T}_l, \mathbf{u}^{(2,1)T}_l,\mathbf{u}^{(3,1)T}_l, \mathbf{u}^{(4,1)T}_l]^T$, $[\mathbf{u}^{(1,2)T}_l, \mathbf{u}^{2T}_l,\mathbf{u}^{(3,2)T}_l, \mathbf{u}^{(4,2)T}_l]^T$, and $[\mathbf{u}^{(1,3)T}_l, \mathbf{u}^{(2,3)T}_l,\mathbf{u}^{3T}_l, \mathbf{u}^{(4,3)T}_l]^T$, we have
\begin{subequations}\label{Equ_K4-4}
\begin{eqnarray}
\mathbf{H}_{1}\mathbf{u}^1_l+\mathbf{H}_{2}\mathbf{u}^{(2,1)}_l+\mathbf{H}_{3}\mathbf{u}^{(3,1)}_l + \mathbf{H}_{4}\mathbf{u}^{(4,1)}_l = \mathbf{0} \label{Equ_K4-4a}\\
\mathbf{H}_{1}\mathbf{u}^{(1,2)}_l+\mathbf{H}_{2}\mathbf{u}^2_l+\mathbf{H}_{3}\mathbf{u}^{(3,2)}_l + \mathbf{H}_{4}\mathbf{u}^{(4,2)}_l = \mathbf{0} \label{Equ_K4-4b}\\
\mathbf{H}_{1}\mathbf{u}^{(1,3)}_l+\mathbf{H}_{2}\mathbf{u}^{(2,3)}_l+\mathbf{H}_{3}\mathbf{u}^3_l + \mathbf{H}_{4}\mathbf{u}^{(4,3)}_l = \mathbf{0}.\label{Equ_K4-4c}
\end{eqnarray}
\end{subequations}
Let $\mathbf{u}^{(1,4)}_l=\mathbf{u}^1_l-\mathbf{u}^{(1,2)}_l -\mathbf{u}^{(1,3)}_l$, $\mathbf{u}^{(2,4)}_l=\mathbf{u}^2_l-\mathbf{u}^{(2,1)}_l +\mathbf{u}^{(2,3)}_l$, and $\mathbf{u}^{(3,4)}_l=\mathbf{u}^3_l-\mathbf{u}^{(3,1)}_l +\mathbf{u}^{(3,2)}_l$.
Subtracting \eqref{Equ_K4-4a} by \eqref{Equ_K4-4b} and \eqref{Equ_K4-4c}, we have
\begin{equation}\label{Equ_K4-5}
\begin{split}
& \mathbf{H}_{1}\mathbf{u}^{(1,4)}_l-\mathbf{H}_{2}\mathbf{u}^{(2,4)}_l-\mathbf{H}_{3}\mathbf{u}^{(3,4)}_l  \\
& +\mathbf{H}_{4}  (\mathbf{u}^{(4,1)}_l-\mathbf{u}^{(4,2)}_l - \mathbf{u}^{(4,3)}_l  ) = \mathbf{0}.
\end{split}
\end{equation}
We now show that each unit spans a subspace of dim-$9$ corresponding to Pattern $2.2$.
Eqs. \eqref{Equ_K4-4} and \eqref{Equ_K4-5} imply that one dimension is saved for the subspace spanned by the signals related to one user. The signal alignment shown in \eqref{Equ_K4-4} and \eqref{Equ_K4-5} is similar to the one used by Pattern $1.3$ shown in \eqref{OneCluster-9}. Then, $\{{\bf H}_k {\bf u}^{(k,k')}_l | \forall k,k', k'\neq k\}$ span a subspace of dim-$9$, while $\mathcal{U}_l\backslash\{\mathbf{h}^{(k,k^\prime)}_l,\mathbf{h}^{(k^\prime,k)}_l\}$ span a subspace of dim-$8$. Hence, we always have a nullspace of dim-$1$ to obtian the linear combination of the signals in each pair at the relay. In this way, we can construct $\frac{4M-N}{3}$ units with Pattern $2.2$, occupying a subspace of dim-$3(4M-N)$ in the relay's signal space.
The remaining $N-3(4M-N)$ dimensions of relay's signal space are used to construct $\frac{N-3(4M-N)}{12}$ units with Pattern $2.1$. Thus, an achievable per-user DoF is given by
\begin{equation}\label{Equ_K4-6}\nonumber
\frac{4M-N}{3}\times \frac{12}{4} + \frac{3(N-3(4M-N))}{12} = M.
\end{equation}
If the overall relay's signal space is occupied by the units with Pattern $2.3$, the maximum achievable per-user DoF of $\frac{N}{9} \times 3 = \frac{N}{3}$ is achieved. Therefore, the achievable per-user DoF is given by $\min(M,\frac{N}{3})=M$ for $\frac{1}{4} < \frac{M}{N} \leq \frac{1}{3}$.

\subsubsection{Case of $\frac{1}{3} < \frac{M}{N} \leq \frac{1}{2}$}
In this case, $\frac{M}{N}$ is large enough to construct the units with Pattern $2.3$ shown in Table~\ref{Table K4}.
We form the following four three-user groups: $\{ {\bf H}_1, {\bf H}_2, {\bf H}_3\}$, $\{ {\bf H}_1, {\bf H}_2, {\bf H}_4\}$, $\{ {\bf H}_2, {\bf H}_3, {\bf H}_4\}$, and $\{ {\bf H}_1, {\bf H}_3, {\bf H}_4\}$, and align the signals within each three-user group.
For each group, the signal alignment is conducted as Patten $1.3$ for $K=3$, where $6$ spatial streams occupy a relay's signal subspace of dim-$4$. Similarly to Pattern $1.3$, $\frac{3M-N}{2}$ units with Pattern $2.3$ are constructed and span a subspace of $4\times \frac{3M-N}{2} = 2(3M-N)$ dimensions. Considering the units from four groups, we have $2(3M-N)$ units which span a subspace of dim-$8(3M-N)$.
The independence of units can be proven by using the result given in Lemma~\ref{Lemma 33} in Appendix~\ref{Appendix 1}.
The remaining $N-8(3M-N)$ dimensions of the relay's signal space are used to construct the units with Pattern $2.2$, the achievable per-user DoF can be expressed as
\begin{equation}\label{Equ_K4-8}\nonumber
\frac{3M-N}{2} \times 3 \times 2 + \frac{3(N-8(3M-N))}{9} = M.
\end{equation}
If the overall relay's signal space is occupied by the units with Pattern $2.3$, we achieve the maximum per-user DoF of $\frac{N}{4} \times \frac{3}{4} \times 2= \frac{3N}{8}$. Hence, the achievable per-user DoF is denoted by $\min(M, \frac{3N}{8})$.

\subsubsection{Case of $\frac{M}{N} > \frac{1}{2}$}
In this case, as $2M-N>0$, the signals in each pair can be aligned in one direction to occupy a subspace of dim-$1$. In total, $2M-N$ units with Pattern $2.4$ can be constructed, which span a subspace of dim-$6(2M-N)$. Similarly, the remaining $N-6(2M-N)$ dimensions of the relay's signal space are used to construct the units with Pattern $2.3$. The achievable per-user DoF can be expressed as
\begin{equation}\label{Equ_K4-9}\nonumber
3(2M-N) + \frac{3(N-6(2M-N))}{8} = \frac{3M}{2} - \frac{3N}{8}.
\end{equation}
If the whole relay's signal space is occupied by the units with Pattern $2.4$, we have the maximum per-user DoF of $\frac{N}{6} \times 3 =\frac{N}{2}$.
The achievable per-user DoF is given by $\min(\frac{3M}{2} - \frac{3N}{8}, \frac{N}{2})$, which is equivalent to
\begin{equation}\label{Equ_K4-10}\nonumber
d_{\mathrm{user}} =
\left\{
\begin{aligned}
& \frac{3M}{2} - \frac{3N}{8}, &\frac{1}{2} <\frac{M}{N} \leq \frac{7}{12}\\
&\frac{N}{2}, &\frac{M}{N} > \frac{7}{12}.
\end{aligned}
\right.
\end{equation}
This completes the proof of Proposition~\ref{Theorem K4}.

\subsection{Achievable DoF for A General $K$}\label{OneCluster-General_K}
We now generalize the achievable DoF result to an arbitrary $K$. Denote
\begin{equation}\label{Equ_alphabeta}
\alpha_{t} =
\left(\begin{smallmatrix}
                                                                                 {K-1} \\
                                                                                 t-1
\end{smallmatrix}
                                                                             \right)
(t-1)~{\rm and}~
\beta_{t} = \left(\begin{smallmatrix}
                                                                                 K \\
                                                                                 t
\end{smallmatrix}
                                                                             \right) (t-1)^2.
\end{equation}
\begin{prop}\label{Theorem K_general}
For the MIMO mRC $(M,N,K)$, the per-user DoF capacity of $d_{{\rm user}}=M$ is achieved when $\frac{M}{N} \in \left(0,\frac{K-1}{K(K-2)} \right]$ and the per-user DoF capacity of $d_{{\rm user}}=\frac{2N}{K}$ is achieved when $\frac{M}{N} \in \left[\frac{1}{K(K-1)}+\frac{1}{2},\infty \right)$.
Further, for the remaining range of $\frac{M}{N}$, an achievable per-user DoF is given by
\begin{equation}\label{Equ_K4-15}
\begin{split}
d_{\rm user} = & \min \bigg( \frac{\alpha_{t}(tM-N)}{t-1} + \frac{\alpha_{t+1}\left(N-\frac{\beta_{t}(tM-N)}{t-1} \right) }{\beta_{t+1}}  , \frac{\alpha_{t}N}{\beta_{t}} \bigg),\\
& \frac{M}{N} \in \left(\frac{1}{t}, \frac{1}{t-1} \right],
\end{split}
\end{equation}
where $t = 2,3,\cdots,K-1$.\footnote{The result given in \eqref{Equ_K4-15} is applicable for a larger range than $\left(\frac{K-1}{K(K-2)}, \frac{1}{K(K-1)}+\frac{1}{2} \right)$ since it includes a union of $K-2$ intervals as $\left(\frac{1}{K-1}, \frac{1}{K-2}\right]\cup \left(\frac{1}{K-2}, \frac{1}{K-3}\right]\cup \cdots \cup \left(\frac{1}{2}, 1\right]$. }
\end{prop}
\begin{proof}
For the case of $\frac{M}{N} < \frac{1}{K}$, it is easy to obtain that $d_{\rm user}=M$, which is also the DoF capacity. On the other hand,
for $\frac{M}{N}\in (1,\infty)$, the achievable DoF is equal to the one achieved at $\frac{M}{N}=1$ as the number of relay antennas $N$ is the bottleneck. Thus, we only focus on the range of $\frac{M}{N} \in (\frac{1}{K}, 1]$ in the following proof. Moreover, we partition the range of $(\frac{1}{K}, 1]$ into intervals of $(\frac{1}{t}, \frac{1}{t-1}]$ with $t=2,3,\cdots,K$, and discuss the signal alignment design for $\frac{M}{N} \in (\frac{1}{t}, \frac{1}{t-1}]$ with an arbitrary $t$.

Note that, $\frac{M}{N} \in (\frac{1}{t}, \frac{1}{t-1}]$ imples $N<tM$,
which indicates that only the spatial streams from $t$ users can be aligned together. Based on that, the most efficient way of signal alignment is to split all the users into different $t$-user groups and perform signal alignment in each group. In this way, we have $\left(
                                                                               \begin{smallmatrix}
                                                                                 {K} \\
                                                                                 t \\
                                                                               \end{smallmatrix}
                                                                             \right)$ number of different $t$-user groups. Further, each user is included in $\left(
                                                                              \begin{smallmatrix}
                                                                                 {K-1} \\
                                                                                 t-1 \\
                                                                               \end{smallmatrix}
                                                                             \right)$ number of different $t$-user groups.
Denote the channel matrices for an arbitrarily chosen $t$-user group as $\left\{{\bf H}_{i_1},{\bf H}_{i_2},\cdots,{\bf H}_{i_t} \right\}$, the nullspace of ${\rm span}\left({\bf H}_{i_1},{\bf H}_{i_2},\cdots,{\bf H}_{i_t} \right)$ is of dim-$(tM-N)$ with probability one.
Similarly to Pattern $2.3$, an efficient way to align signals in one unit is let the spatial streams related to one user be aligned together (without loss of generality, we term this pattern as Pattern $3.t$). The beamformers in unit $l$ can be designed to satisfy the conditions given in \eqref{Equ_K4-11-a}-\eqref{Equ_K4-11-c} shown at the top of next page,
\newcounter{LangEq}
\begin{figure*}[!t]
\normalsize \setcounter{LangEq}{\value{equation}}
\setcounter{equation}{21}
\begin{subequations}\label{Equ_K4-11}
\begin{eqnarray}
 {\bf H}_{i_1} \sum^{t}_{i=2} a^{(1,i)}_l {\bf u}^{(1,i)}_l  + {\bf H}_{i_2} {\bf u}^{(2,1)}_l + \cdots + {\bf H}_{{i_{t-1}}} {\bf u}^{(t-1,1)}_l + {\bf H}_{i_t} {\bf u}^{(t,1)}_l = {\bf 0}, \label{Equ_K4-11-a}\\
 {\bf H}_{i_1} {\bf u}^{(1,2)}_l +  {\bf H}_{i_2} \sum^{t}_{i=1, i \neq 2} a^{(2,i)}_l {\bf u}^{(2,i)}_l  + \cdots + {\bf H}_{{i_{t-1}}} {\bf u}^{(t-1,2)}_l + {\bf H}_{i_t} {\bf u}^{(t,2)}_l = {\bf 0}, \label{Equ_K4-11-b}\\
\vdots  ~~~~~~~~~~~~~~~~~ \nonumber \\
 {\bf H}_{i_1} {\bf u}^{(1,t-1)}_l +  {\bf H}_{i_2} {\bf u}^{(2,t-1)}_l   + \cdots + {\bf H}_{{i_{t-1}}} \sum^{t}_{i=1, i \neq t-1} a^{(t-1,i)}_l {\bf u}^{(t-1,i)}_l  + {\bf H}_{i_t} {\bf u}^{(t,t-1)}_l = {\bf 0}, \label{Equ_K4-11-c} \\
  {\bf H}_{i_1}  {\bf u}^{(1,t)}_l  + {\bf H}_{i_2}  {\bf u}^{(2,t)}_l + \cdots + {\bf H}_{{i_{t-1}}} {\bf u}^{(t-1,t)}_l + {\bf H}_{i_t} \left({\bf u}^{(t,1)}_l -  \sum^{t-1}_{i=2} {\bf u}^{(t,i)}_l \right) = {\bf 0}. \label{Equ_K4-11-d}
\end{eqnarray}
\end{subequations}
\setcounter{equation}{\value{LangEq}}
\hrulefill
\end{figure*}
where
\begin{equation}\label{New_aij}\nonumber
a^{(i,j)}_l =
\left\{
\begin{aligned}
& 1, & i=1~{\rm or}~j=1\\
& -1, & {\rm otherwise}.
\end{aligned}
\right.
\end{equation}
Subtracting \eqref{Equ_K4-11-a} by equations from \eqref{Equ_K4-11-b} to \eqref{Equ_K4-11-c}, we obtain the equation given in \eqref{Equ_K4-11-d} shown at the top of next page.
Based on \eqref{Equ_K4-11}, we see that total $t-1$ dimensions can be saved for the subspace spanned by the signals in unit $l$.
In total, we can construct $\frac{tM-N}{t-1}$ units with Pattern $3.t$ in each group and the achievable per-user DoF in one unit is $t-1$. Thus, the achievable per-user DoF is $ \frac{\alpha_t(tM-N)}{t-1}$.
Note that the linear independence of the subspaces spanned by the units with Pattern $3.t$ can be proven using the result in Lemma~\ref{Lemma 33} in Appendix~\ref{Appendix 1}.
Moreover, for each $t$-user group, one unit spans a subspace of $t(t-1) -(t-1) = (t-1)^2$ dimensions. Considering all $\left(
                                                                               \begin{smallmatrix}
                                                                                 K \\
                                                                                 t \\
                                                                               \end{smallmatrix}
                                                                             \right)$ number of groups, all the units span a subspace of $  \frac{\beta_t(tM-N)}{t-1}$ dimensions.
The remaining $N- \frac{\beta_t(tM-N)}{t-1}$ dimensions of relay's signal space are used to construct the units with Pattern $3.(t+1)$. Then, the achievable per-user DoF can be expressed as
\setcounter{equation}{22}
\begin{equation}\label{Equ_K4-17}
d_{\rm user} =  \frac{\alpha_t(tM-N)}{t-1} + \frac{\alpha_{t+1}\left(N-  \frac{\beta_{t}(tM-N)}{t-1}\right)}{\beta_{t+1}} .
\end{equation}
When the relay's signal space is wholly occupied by units with Pattern $3.t$, we achieve a maximum achievable per-user DoF of $\frac{\alpha_tN}{\beta_t}$. Thus, an achievable per-user DoF is obtained as \eqref{Equ_K4-15}.

When $t=K$, we have only one group, which leads to $\alpha_K = K-1$ and $\beta_K=(K-1)^2$. Note that for $\frac{M}{N}<\frac{1}{K}$, we cannot do any signal alignment, $\alpha_{K+1}$ and $\beta_{K+1}$ defined in \eqref{Equ_K4-17} should be equal to $K-1$ and $K(K-1)$, respectively.
Substitute them into \eqref{Equ_K4-15}, we have $d_{\rm user}=\min (M, \frac{N}{K-1})=M$, which is the DoF capacity.

When $t=K-1$, we have $\alpha_{K-1} = (K-1)(K-2)$ and $\beta_{K-1}=K(K-2)^2$. Substitute $\alpha_K$, $\beta_K$, $\alpha_{K-1}$, and $\beta_{K-1}$ into \eqref{Equ_K4-15}, we obtain $d_{\rm user}=\min \left(M, \frac{(K-1)N}{K(K-2)} \right)=M$, which implies that the per-user DoF capacity of $d_{\rm user}=M$ can also be achieved for $\frac{M}{N} \in \left(\frac{1}{K-1},\frac{K-1}{K(K-2)} \right]$.

Similarly, we can also verify that the per-user DoF capacity of $\frac{2N}{K}$ can be achieved in $\frac{M}{N} \in \left[\frac{1}{K(K-1)}+\frac{1}{2},1 \right]$ by letting $t=2$. Then, we complete the proof Proposition~\ref{Theorem K_general}.
\end{proof}

When the number of the users $K$ tends to infinity, the following asymptotic DoF can be obtained from Proposition~\ref{Theorem K_general}.\footnote{When $K\rightarrow \infty$, $d_{\rm user}$ tends to $0$. Thus, we consider the total achievable DoF $d_{\rm sum}$ here.}

\begin{figure}[tp]
\begin{centering}
\includegraphics[scale=0.50]{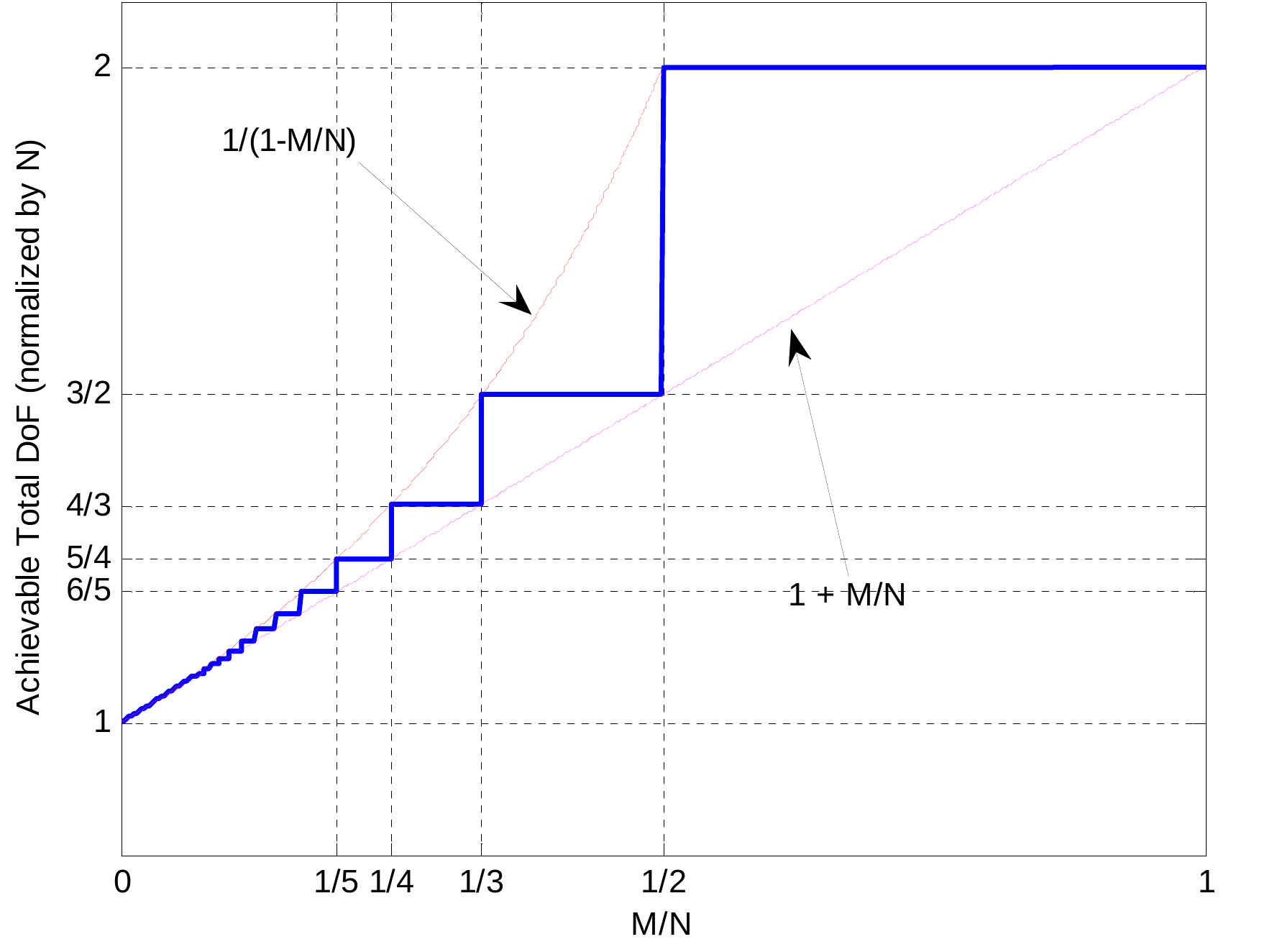}
\vspace{-0.1cm}
\caption{A total achievable DoF of the MIMO mRC with $K\rightarrow\infty$ against the antenna ratio $\frac{M}{N}$.}  \label{K_infty}
\end{centering}
\vspace{-0.3cm}
\end{figure}

\begin{corollary}\label{Theorem K_infinite}
For the MIMO mRC $(M,N,K)$ with $K\rightarrow \infty$, the total DoF of $d_{\rm sum}=2N$ is achieved when $\frac{M}{N} > \frac{1}{2}$ and $d_{\rm sum}=N$ is achieved as $\frac{M}{N}\rightarrow 0$.
For $\frac{M}{N} \in (0, \frac{1}{2}]$, the achievable total DoF contains discontinuities at $\frac{M}{N} = \frac{1}{t}, t=2,3,4,\cdots$. Specifically, when $\frac{M}{N} \in (\frac{1}{t}, \frac{1}{t-1}]$, a total DoF of $\frac{tN}{t-1}$ is achieved.
\end{corollary}

\emph{Remark} 1:
When $K \rightarrow \infty$, the number of spatial streams $K(K-1)$ tends infinity. The total achievable DoF is then bounded by the number of relay antennas $N$, which further implies that only a portion of users can realize data exchange.
The overall achievable total DoF with respect to $\frac{M}{N}$ is illustrated in Fig.~\ref{K_infty}. From Corollary~\ref{Theorem K_infinite}, we see that for each antenna setup $\frac{M}{N} = \frac{1}{t}$, the total achievable DoF jumps from $\frac{t+1}{t}$ to $\frac{t}{t-1}$. It is interesting to verify that the discontinuous points $(\frac{M}{N},\frac{d_{\rm sum}}{N})=(\frac{1}{t}, \frac{t+1}{t})$ are went through by the line $y = 1 + M/N$, while the discontinuous points $(\frac{M}{N},\frac{d_{\rm sum}}{N})=(\frac{1}{t}, \frac{t}{t-1})$ are enveloped by the curve $y = \frac{1}{1 - M/N}$. In this case, the achievable total DoF is nicely bounded by these two curves as shown in Fig.~\ref{K_infty}. Therefore, without loss of generality, we refer to $ \frac{N^2}{N - M}$ as an upper bound of the total achievable DoF and refer to $N + M$ as a lower bound of the total achievable DoF.

\section{Improved Achievable DoF Using Relay Antenna Deactivation}
In the previous sections, we have shown that the proposed beamforming design achieves the DoF capacity for the MIMO mRC with $K=3$. However, a certain gap occurs in the range of $\frac{M}{N}\in \left(\frac{K-1}{K(K-2)}, \frac{1}{K(K-1)}+\frac{1}{2}\right)$ when $K>3$.
In this section, we show that the obtained achievable DoF in Proposition~\ref{Theorem K_general} can be enhanced by the technique of relay antenna deactivation, i.e., to leave a portion of relay antennas disabled in the uplink and downlink transmissions.
We emphasize that the antenna deactivation technique in general cannot improve the DoF of the considered relay channel. The improvement presented below comes from the non-optimality of the signal alignment technique utilized in Section V.

To proceed, we first give the following property which can be directly obtained from Proposition~\ref{Theorem K_general}.
\begin{property}\label{property_1}
For a MIMO mRC with an antenna configuration of $(M,N)$, the obtained achievable DoF in Proposition~\ref{Theorem K_general} can always be represented by $\psi(\frac{M}{N})N$ where $\psi(\frac{M}{N})$ is a coefficient determined solely by $\frac{M}{N}$.
\end{property}

Property~\ref{property_1} implies that when the antenna configuration of a MIMO mRC varies from $(M,N)$ to $(\sigma M,\sigma N)$ where $\sigma>0$ is an arbitrary coefficient, the obtained achievable DoF by Proposition~\ref{Theorem K_general} changes from $\psi(\frac{M}{N}) N$ to $\psi(\frac{M}{N})\sigma N$. Then, we have the following lemma.

\begin{lemma}\label{lemma_1}
For the MIMO mRC $(M,N,K)$, assume that a certain DoF of $d_{\rm user}=\varphi N_0$ is achievable at $(M,N)=(M_0, N_0)$ where $M_0$ and $N_0$ are certain constant integers, and $\varphi \in (0,1]$ is a coefficient. Then $d_{\rm user}=\varphi N_0 \frac{M}{M_0}$ is achievable for any $M\leq M_0$ by disabling a fraction $1-\frac{M}{M_0}$ of all the relay antennas.
\end{lemma}
\begin{proof}
Consider an antenna setup of $(M,N_0)$ with $M<M_0$.  As $\frac{M}{N_0}<\frac{M_0}{N_0}$, we can reduce the number of active relay antennas to $N^\prime=\frac{MN_0}{M_0}$ by disabling a fraction $1-\frac{M}{M_0}$ of all the relay antennas.\footnote{Here, we assume that $\frac{MN_0}{M_0}$ is an integer. Otherwise, the technique of symbol extension should be used to ensure that the number of disabled relay antennas is an integer.} Then, we have $\frac{M}{N^\prime}=\frac{M_0}{N_0}$.
Based on Property~\ref{property_1}, a DoF of $d_{\rm user}=\varphi N^\prime=\varphi N_0\frac{M}{M_0}$ can be achieved.
This completes the proof of Lemma~\ref{lemma_1}.
\end{proof}

Lemma~\ref{lemma_1} implies that for each achievable DoF $d$ (normalized by $N$) obtained in Proposition~\ref{Theorem K_general} at $\frac{M}{N}=a$, we obtain a line segment $y=\frac{d}{a}x$ containing new achievable DoF (normalized by $N$) in the range of $\frac{M}{N}\in(0,a]$. Take Fig.~\ref{Im_L_1_K_4} as an example. Consider that the point $(\frac{M}{N},\frac{d_{\rm user}}{N})=(\frac{7}{12},\frac{1}{2})$ is achievable.\footnote{Here, we say that a point $(\frac{M}{N},\frac{d_{\rm user}}{N})=(\vartheta_0,d_0)$ is achievable if a DoF of $d_0$ (normalized by $N$) is achieved at $\frac{M}{N}=\vartheta_0$.} Then, based on Lemma~\ref{lemma_1}, we obtain that all the points on the line segment connecting $(\frac{M}{N},\frac{d_{\rm user}}{N})=(0,0)$ and $(\frac{M}{N},\frac{d_{\rm user}}{N})=(\frac{7}{12},\frac{1}{2})$, i.e., $y=\frac{12}{14}x=\frac{6}{7}x$ with $x\in (0,\frac{7}{12}]$, are achievable.

We next improve the results in Proposition~\ref{Theorem K_general} using Lemma~\ref{lemma_1}.
To simplify the notation in \eqref{Equ_K4-15}, we denote
\begin{equation}\label{Equ_Imporve_1}
\begin{split}
 \gamma_{t,1} &= \frac{\alpha_t(tM-N)}{t-1} + \frac{\alpha_{t+1}\left(N-  \frac{\beta_t(tM-N)}{t-1}\right) }{\beta_{t+1}} \\
 \gamma_{t,2} &= \frac{\alpha_t N}{\beta_t}\\
 \theta_t & = \frac{t-1}{t \beta_t}+ \frac{1}{t},~t=2,3,\cdots,K-1.
\end{split}
\end{equation}
Note that $\gamma_{t,1}$ is the achievable DoF when the overall relay's signal space is occupied jointly by units following Pattern $3.t$ (see \eqref{Equ_K4-11}) and Pattern $3.(t+1)$, whereas $\gamma_{t,2}$ is the achievable DoF when the overall relay's signal space is occupied by Pattern $3.t$.
Particularly, for $\frac{M}{N}=\theta_t$, we have $\gamma_{t,1}=\gamma_{t,2}$, implying that the overall relay's signal space is occupied by only one pattern, i.e., Pattern $3.t$. In what follows, we refer to the point $(\frac{M}{N},\frac{d_{\rm user}}{N})=(\theta_t, \frac{\gamma_{t,2}}{N})$ ($\frac{\gamma_{t,2}}{N}$ is the normalized achievable DoF) as a corner point. For $K=4$,  $(\frac{M}{N},\frac{d_{\rm user}}{N})=(\frac{7}{12}, \frac{1}{2})$ is a \emph{corner point} as illustrated in Fig.~\ref{L_1_K_4}. Next, by applying Lemma~\ref{lemma_1}, we use the corner point $(\frac{M}{N},\frac{d_{\rm user}}{N})=(\theta_t, \frac{\gamma_{t,2}}{N})$ to improve the achievable DoF in the interval of $\frac{M}{N} \in (\theta_{t+1},\theta_t]$ in Proposition~\ref{Theorem K_general}, as presented in the following lemma.

\begin{lemma}\label{lemma_2}
For $\frac{M}{N} \in (\theta_{t+1},\theta_t]$, the improved achievable per-user DoF is given by $\frac{Mt \alpha_t}{t-1 + \beta_t}$ in the range of $\frac{M}{N} \in (\frac{ \alpha_{t+1}(t-1 + \beta_t)}{t\alpha_t \beta_{t+1}}, \theta_t)$.
\end{lemma}

Lemma~\ref{lemma_2} implies that the results in Proposition~\ref{Theorem K_general} can be improved in the interval of $\frac{M}{N}\in\left(\theta_{t}, \theta_{t-1}\right]$ with $t=3,4,\cdots,K-1$ by disabling a protion of relay antennas.
Take $K=4$ as an example, as illustrated in Fig.~\ref{Im_L_1_K_4}. We see that, in the range of $\frac{M}{N}\in(\theta_{3},\theta_{2})$, i.e., $\frac{M}{N}\in(\frac{3}{8},\frac{7}{12})$, the original achievable DoF is improved for $\frac{M}{N}\in(\frac{7}{16},\frac{7}{12})$, and the normalized improved achievable DoF is on the line segment of $y=\frac{\gamma_{2,2}/N}{\theta_2}x=\frac{12}{14}x=\frac{6}{7}x$.

Based on Lemma~\ref{lemma_2}, we have the following theorem.

\begin{theorem}\label{Imporve K_general}
For the MIMO mRC $(M,N,K)$, the per-user DoF capacity of $d_{\rm user} = M$ is achieved when $\frac{M}{N} \in \left(0,\frac{K-1}{K(K-2)} \right]$ and the per-user DoF capacity of
$d_{\rm user} = \frac{2N}{K}$ is achieved when $\frac{M}{N} \in \left[\frac{1}{K(K-1)}+\frac{1}{2},\infty \right)$.
For $\frac{M}{N} \in \left(\frac{K-1}{K(K-2)}, \frac{1}{K(K-1)}+\frac{1}{2} \right)$, an achievable per-user DoF is given by
\begin{equation}\nonumber
d_{\mathrm{user}} =
\left\{
\begin{aligned}
& \frac{N \alpha_{t+1}}{\beta_{t+1}}  & \frac{M}{N} \in \left(\frac{t}{(t+1) \beta_{t+1}}+ \frac{1}{t+1}, \tau_t \right]\\
&\frac{M t \alpha_{t}}{t-1 + \beta_{t}}, & \frac{M}{N} \in \left(\tau_t, \frac{t-1}{t \beta_{t}}+ \frac{1}{t}\right],
\end{aligned}
\right.
\end{equation}
where $\alpha_{t}$, $\beta_{t}$ are defined in \eqref{Equ_alphabeta}, $\tau_t=\frac{ \alpha_{t+1}(t-1 + \beta_{t})}{t \alpha_{t} \beta_{t+1}}$,
and  $t\in [2, K-2]$ is an integer.
\end{theorem}

\begin{figure}[tp]
\begin{centering}
\includegraphics[scale=0.5]{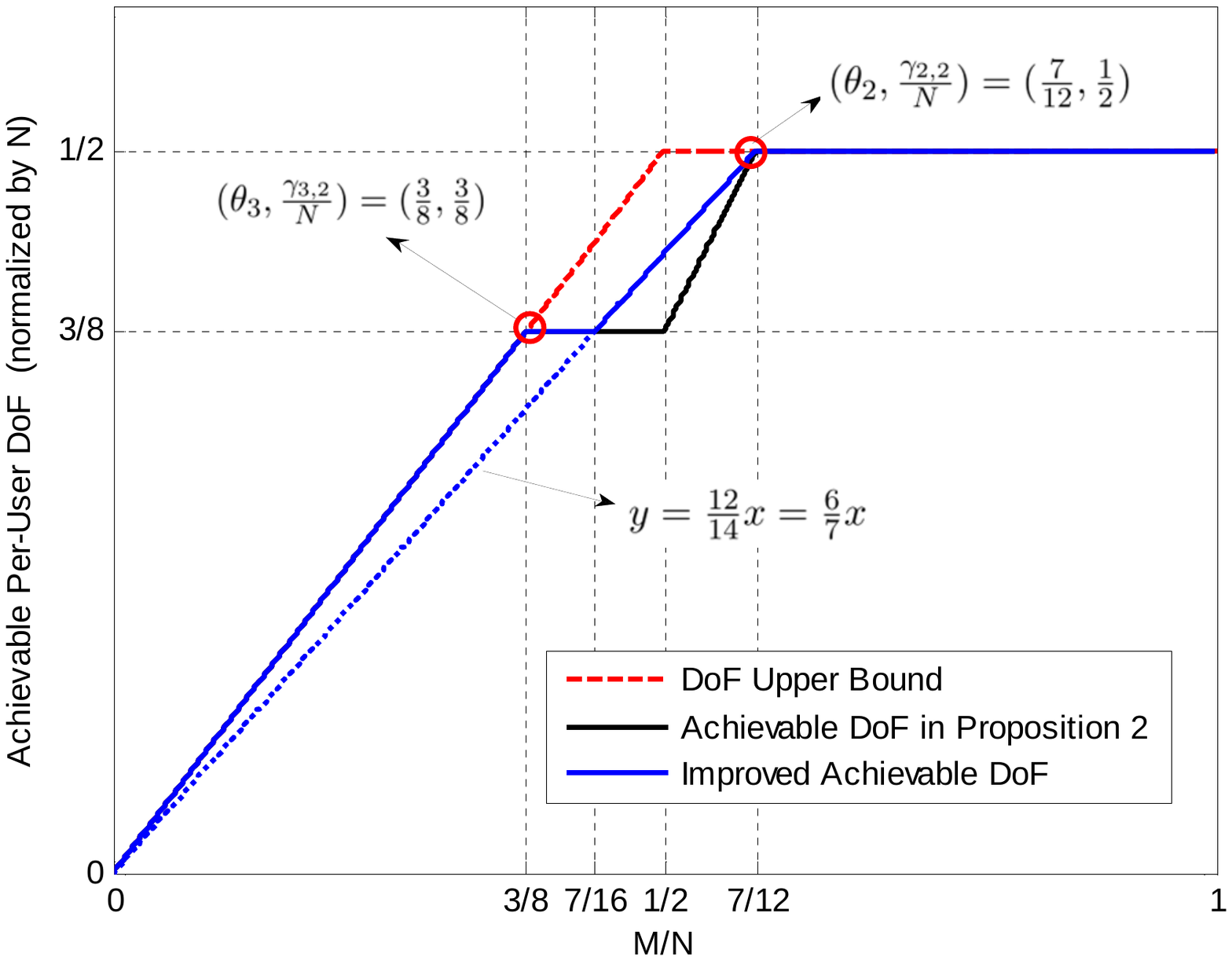}
\vspace{-0.1cm}
\caption{The improved achievable per-user DoF for MIMO mRC with $K=4$ against the antenna ratio $\frac{M}{N}$.}  \label{Im_L_1_K_4}
\end{centering}
\vspace{-0.3cm}
\end{figure}

\begin{corollary}\label{improved K_4}
For the MIMO mRC $(M,N,K)$ with $K=4$, the per-user DoF capacity of $d_{\rm user}=M$ is achieved when $0 < \frac{M}{N} \leq \frac{3}{8}$,
and the per-user DoF capacity of $d_{\rm user} =\frac{N}{2}$ is achieved when $\frac{M}{N} \geq \frac{7}{12}$. For $\frac{M}{N} \in (\frac{3}{8}, \frac{7}{12})$,
an achievable per-user DoF is given by
\begin{equation}\label{DoF_K4}
d_{\mathrm{user}} =
\left\{
\begin{aligned}
& \frac{3N}{8} &\frac{3}{8} < \frac{M}{N} \leq \frac{7}{16}\\
&\frac{6M}{7}, &\frac{7}{16} < \frac{M}{N} \leq \frac{7}{12}.\\
\end{aligned}
\right.
\end{equation}
\end{corollary}

Corollary~\ref{improved K_4} is obtained from Theorem~\ref{Imporve K_general} by letting $K=4$, with the DoF curve illustrated in Fig.~\ref{Im_L_1_K_4}.
It is interesting to see that the obtained achievable DoF curve is piecewise linear, depending on the number of user antennas $M$ and the number of relay antennas $N$ alternately. This result is similar to the DoF capacity of the interference channel obtained in \cite{Wang_IT2011}. The piecewise linearity implies antenna redundancy. Specifically, for $\frac{M}{N} \in (0, \frac{3}{8}] \cup (\frac{7}{16}, \frac{7}{12}) $, the derived achievable DoF is bounded by the number of user antennas, implying antenna redundancy at the relay; for $\frac{M}{N} \in (\frac{3}{8},  \frac{7}{16}] \cup (\frac{7}{12},1] $, the derived achievable DoF is bounded by the number of relay antennas, which implies that the antenna redundancy occurs at the user side.

It is worth noting that the author in \cite{Wang_IT2014} (which is parallel to and independent of the work in this paper) derived the DoF for the MIMO Y channel with $K=4$. The result in \cite{Wang_IT2014} slightly outperforms the achievable DoF in Corollary~\ref{improved K_4} in the interval of $\frac{M}{N}\in \left(\frac{3}{8},\frac{1}{2}\right)$. This implies that the proposed scheme in this paper is in general suboptimal for $K\geq 4$. As the signal alignment technique in \cite{Wang_IT2014} is limited to the case of $K=4$, the DoF capacity of the MIMO Y channel with an arbitrary $K$ still remains a challenging open problem worthy of future research endeavor.

When the user number $K\rightarrow \infty$, we obtain an asymptotic DoF given in the following corollary.
\begin{corollary}\label{improved K_infinite}
For the MIMO mRC $(M,N,K)$ with $K \rightarrow\infty$, the total DoF of $d_{\rm sum}=2N$ is achieved when $\frac{M}{N} > \frac{1}{2}$ and $d_{\rm sum}=N$ is achieved as $\frac{M}{N}\rightarrow 0$.
For $\frac{M}{N} \in (0, \frac{1}{2}]$, the achievable total DoF is piecewise linear with respect to either $M$ or $N$.
Specifically, for $t=2,3,\cdots,\infty$,
we have $d_{\rm sum}=\frac{(t+1)N}{t}$ for $\frac{M}{N} \in \left(\frac{1}{t+1}, \frac{(t+1)(t-1)}{t^3} \right]$, and $d_{\rm sum}=\frac{M t^2}{t-1}$ for $\frac{M}{N} \in \left(\frac{(t+1)(t-1)}{t^3}, \frac{1}{t} \right]$.
\end{corollary}

The DoF curve for Corollary~\ref{improved K_infinite} is illustrated in Fig.~\ref{Im_K_infty}. We see that the normalized achievable DoF for $\frac{M}{N}\in (0,\frac{1}{2}]$ is enveloped by the curve of $y=\frac{1}{1-M/N}$.
Further, the range of $ (0,\frac{3}{8}]$ is partitioned into an infinite number of intervals, namely, $ \left(\frac{(t+2)t}{(t+1)^3}, \frac{(t+1)(t-1)}{t^3} \right] $ for $t=2,3,\cdots,\infty$. A new pattern arises for efficient signal alignment when $\frac{M}{N}$ moves to a new interval.
Similarly, the achievable DoF in the ranges of $\frac{M}{N}\in \left(\frac{(t+2)t}{(t+1)^3},\frac{1}{t+1} \right]$ and $\frac{M}{N}\in \left(\frac{1}{t+1}, \frac{(t+1)(t-1)}{t^3} \right] $ implies antenna redundancy at the user side and at the relay side, respectively.

\begin{figure}[tp]
\begin{centering}
\includegraphics[scale=0.5]{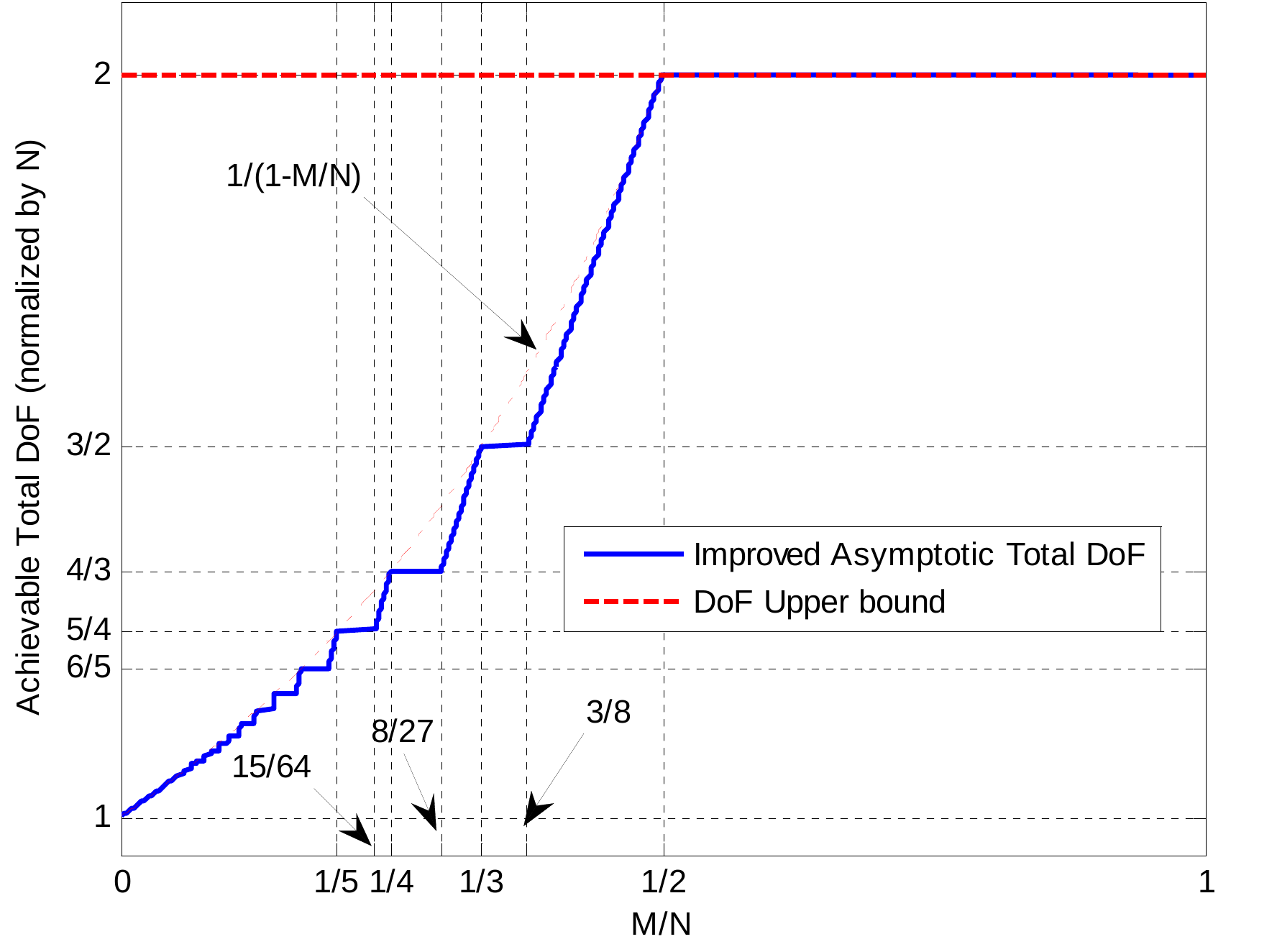}
\vspace{-0.1cm}
\caption{The improved total achievable DoF for the  MIMO mRC with $K\rightarrow \infty$ against the antenna ratio $\frac{M}{N}$.}  \label{Im_K_infty}
\end{centering}
\vspace{-0.3cm}
\end{figure}

\emph{Remark} 2: Before leaving this section, we provide some further comments on the antenna deactivation technique used in Lemma~\ref{lemma_1}.
In general, disabling a portion of relay antennas will reduce the DoF capacity of the considered channel. However, in some cases,
the obtained achievable DoF by relay antenna disablement coincides with the DoF capacity of the channel. For example, for $K=4$ as shown in Fig.~\ref{Im_L_1_K_4}, our result for $\frac{M}{N}\in \left[\frac{1}{2},\frac{7}{12} \right)$ is in fact the DoF capacity of the considered relay channel; see \cite{Wang_IT2014}. An intuitive explanation is that for $\frac{M}{N}\in \left[\frac{1}{2},\frac{7}{12} \right)$, the number of antennas at each user end is the performance bottleneck of the channel, which implies antenna redundance at the relay. That means, disabling a certain portion of relay antennas does not incur any DoF penalty in this case.

\section{Conclusion}

In this paper, we studied an achievable DoF of the MIMO mRC for an arbitrary number of users with any antenna setups. A novel and systematic way of beamforming design was proposed to realize different kind of signal space alignments, which were then used to implement PNC. It was shown that the proposed signal alignment scheme achieves the DoF capacity of the MIMO mRC with $K=3$. For the case of $K>3$, we showed that our proposed signal alignment scheme achieves the DoF capacity of the MIMO mRC for $\frac{M}{N} \in \left(0,\frac{K-1}{K(K-2)} \right]$ and $\frac{M}{N} \in \left[\frac{1}{K(K-1)}+\frac{1}{2},\infty \right)$. This result has a broader range of $\frac{M}{N}$ compared to the existing achievability of the DoF capacity in \cite{Tian_IT2013}.
The asymptotic achievable DoF  when the number of users tending to infinity was also analyzed. The derived achievable DoF in this work can in general serve as a lower bound of the DoF capacity of the considered MIMO mRC. Future research interests include the optimal precoding design of the MIMO mRC in finite SNR and the extension of our DoF analysis to more complicated scenarios, such as clustered MIMO mRCs.

\appendices
\section{Some Useful Lemmas}\label{Appendix 1}
Let $\mathbf{A}_i \in \mathbb{C}^{N\times M}$, for $i = 1, 2, \cdots, K$, be independent random matrices with $M \leq N$.
Assume $KM > N$. Denote $\mathbf{A} = [\mathbf{A}_1,\cdots, \mathbf{A}_K]$, and let the columns of $\mathbf{U} \in C^{KM\times (KM-N)}$ be a basis of $\mathrm{null}(\mathbf{A})$. Further, we represent $\mathbf{U}$ as $\mathbf{U} = [\mathbf{U}_1^T, \cdots, \mathbf{U}_K^T]^T$, where $\mathbf{U}_i\in C^{M\times (KM-N)}$. Denote $\mathbf{\tilde{A}} = [\mathbf{A}_1\mathbf{U}_1,\cdots, \mathbf{A}_K\mathbf{U}_K] \in C^{N\times K(KM-N)}$. We have the following result.
\begin{lemma}\label{Lemma 4} The rank of matrix $\mathbf{\tilde{A}}$ is $\min\left((K-1)(KM-N),N\right)$ with probability one.
\end{lemma}
\begin{proof}
We first assume $(K-1)(KM-N) \leq N$, or equivalently, $(K-1)M \leq N$. Let $\mathbf{v} = [\mathbf{v}_{1}^T, \cdots, \mathbf{v}_{K}^T]^T$ be an arbitrary vector in $\mathrm{null}(\mathbf{\tilde{A}})$, where $\mathbf{v}_{i} \in \mathbb{C}^{(KM-N)\times 1}$. Then, we obtain $\mathbf{0} = \mathbf{\tilde{A}}\mathbf{v} = \mathbf{A}\mathrm{diag}\left(\mathbf{U}_{1}, \cdots, \mathbf{U}_{K}\right)\mathbf{v}$. Thus, $\mathrm{diag}\left(\mathbf{U}_{1}, \cdots, \mathbf{U}_{K}\right)\mathbf{v}$ belongs to $\mathrm{null}(\mathbf{A})$. As $\mathbf{U}$ spans $\mathrm{null}(\mathbf{A})$, there exists $\mathbf{x} \in \mathbb{C}^{(KM-N)\times 1}$ such that $\mathrm{diag}\left(\mathbf{U}_{1}, \cdots, \mathbf{U}_{K}\right)\mathbf{v} = \mathbf{Ux}$, or equivalently, $\mathbf{U}_{i}\mathbf{v}_{i} = \mathbf{U}_{i}\mathbf{x}$, for $\forall i$. From the randomness of $\mathbf{A}$, $\mathbf{U}_{i}$ is of rank $\mathrm{min}(M, KM-N)$ with probability one. By assumption of $(K-1)(KM-N) \leq N$, we obtain $KM-N \leq M$. Thus, the left inverse of $\mathbf{U}_{i}$ exists with probability one. Then, $\mathbf{U}_{i}\mathbf{v}_{i} = \mathbf{U}_{i}\mathbf{x}$ implies $\mathbf{v}_{i} = \mathbf{x}$, for $\forall i$, or equivalently, $\mathbf{v} = [\mathbf{x}^T,\cdots, \mathbf{x}^T]^T$. Since $\mathbf{x}$ can be any vector in $\mathbb{C}^{(KM-N)\times 1}$, the nullspace of $\mathbf{\tilde{A}}$ is of rank $KM-N$ with high probability. By using the rank-nullity theorem of linear algebra, we conclude that $\mathbf{\tilde{A}}$ is of rank $K(KM-N)-(KM-N) = (K-1)(KM-N)$ with probability one. What remains is the case of $(K-1)(KM-N) > N$. Note that, when $(K-1)(KM-N) = N$, $\mathrm{span}(\mathbf{\tilde{A}})$ is of dimension $N$ with probability one. Further increasing $M$ cannot increase the dimension of $\mathrm{span}(\mathbf{\tilde{A}})$ (as $\mathbf{\tilde{A}}$ only has $N$ rows), which concludes the proof.
\end{proof}

\begin{lemma}\label{Lemma 3}
The subspace of $\mathrm{span}(\mathbf{A}_1)\cap\mathrm{span}(\mathbf{A}_2)$, i.e., the intersection of $\mathrm{span}(\mathbf{A}_1)$ and $\mathrm{span}(\mathbf{A}_2)$, has a dimension of $(2M-N)^+$ with probability one.
\end{lemma}
\begin{proof}
This result has been proven in \emph{Lemma 1} of \cite{Lee10}. We omit the details here for brevity.
\end{proof}

For $\frac{M}{N}\in \big(\frac{1}{t},\frac{1}{t-1}\big]$, denote by
$\{j_1,j_2,\cdots,j_t\}$ $t$ distinct indexes chosen from $1,2,\cdots,K$ with $j_1<j_2<\cdots<j_t$.
Let ${\bf A}_{[j]}=[{\bf A}_{j_1},{\bf A}_{j_2},\cdots,{\bf A}_{j_t}]$.
We say that ${\bf A}_{[j]}$ is different from ${\bf A}_{[j^\prime]}$ for
any $j\neq j^\prime$ if there is at least one ${\bf A}_i$ in ${\bf A}_{[j]}$ not contained in ${\bf A}_{[j^\prime]}$. In total, we have $J=\left(
                                                                               \begin{smallmatrix}
                                                                                 {K} \\
                                                                                 t \\
                                                                               \end{smallmatrix}
                                                                             \right)$ different choices of ${\bf A}_{[j]}$.
Let the columns of ${\bf U}_{[j]}=[\mathbf{U}_{j_1}^T, \mathbf{U}_{j_2}^T,\cdots, \mathbf{U}_{j_t}^T]^T$ be a basis of ${\rm null}({\bf A}_j)$ with $\mathbf{U}_{j_i}\in C^{M\times (tM-N)}$ for $\forall i$, and $\mathcal{S}_j = {\rm span}\left(\{{\bf A}_{j_i}{\bf U}_{j_i}|i=1,2,\cdots,t\}\right)$. Then

\begin{lemma}\label{Lemma 33}
$\mathcal{S}_1 \oplus \mathcal{S}_2 \oplus \cdots \oplus \mathcal{S}_J $ is a subspace of dimension of $\mathrm{min}\left(J(t-1)(tM-N), N \right)$ with probability one.
\end{lemma}
\begin{proof}
We first focus on the case of $N\geq J(t-1)(tM-N)$.
Denote $\tilde{{\bf A}}_{[j]}=[{\bf A}_{j_1}{\bf U}_{j_1},{\bf A}_{j_2}{\bf U}_{j_2},\cdots,{\bf A}_{j_{t-1}}{\bf U}_{j_{t-1}}]$ and $\tilde{{\bf A}}=[\tilde{{\bf A}}_{[1]},\tilde{{\bf A}}_{[2]},\cdots,\tilde{{\bf A}}_{[J]} ]$.
By noting ${\bf A}_{[j]} {\bf U}_{[j]}=\sum^t_{i=1}{\bf A}_{j_i}{\bf U}_{j_i}={\bf 0}$,
we have $\mathcal{S}_j={\rm span}(\tilde{{\bf A}}_{[j]})$.
To prove Lemma \ref{Lemma 33}, it suffices to show that $\tilde{{\bf A}} \in \mathbb{C}^{N\times{J(t-1)(tM-N)}}$ is of full column rank. For notation convenience, we focus on the case of $K=4$ and $t=3$. The proof can be readily extended to the case of an arbitrary $K$ and $t$.

For $K=4$, in total, we have $4$ different choices of ${\bf A}_{[j]}=[{\bf A}_{j_1},{\bf A}_{j_2},{\bf A}_{j_3}]$.
Let $\mathcal{S}_{j1} =\mathrm{span}({\bf A}_{j_3}{\bf U}_{j_3})= \mathrm{span}(\mathbf{A}_{j_1},\mathbf{A}_{j_2}) \cap \mathrm{span}(\mathbf{A}_{j_3}) $, $\mathcal{S}_{j2} =\mathrm{span}({\bf A}_{j_2}{\bf U}_{j_2})= \mathrm{span}(\mathbf{A}_{j_1},\mathbf{A}_{j_3}) \cap \mathrm{span}(\mathbf{A}_{j_2}) $.
Then we have $\mathcal{S}_{j}=\mathcal{S}_{j1}\oplus \mathcal{S}_{j2}$, for $j=1,2,3,4$. Also, we have $\mathcal{S}_{j1}\cap \mathcal{S}_{j2}=\{0\}$ by noting the channel randomness and the fact of $2(3M-N)<N$ (implied by $N\geq J(t-1)(3M-N)$).
Consider an $N\times 1$ vector $\mathbf{a}_{j1} \in \mathcal{S}_{j1}$, for $j=1,2,3,4$. By definition of $\mathcal{S}_{j1}$, we have $\mathbf{a}_{j1}\in \mathrm{span}(\mathbf{A}_{j_1},\mathbf{A}_{j_2}) $ and $\mathbf{a}_{j1}\in \mathrm{span}(\mathbf{A}_{j_3}) $.
Thus, there exist unique $\mathbf{q}^1_{j1}$, $\mathbf{q}^2_{j1}$  and $\mathbf{q}^3_{j1}$ such that
\begin{equation}\nonumber
\mathbf{a}_{j1} = \mathbf{A}_{j_1} \mathbf{q}^1_{j1} +  \mathbf{A}_{j_2} \mathbf{q}^2_{j1} = \mathbf{A}_{j_3} \mathbf{q}^3_{j1},
\end{equation}
or equivalently
\begin{equation}\label{Eq1}
{\bf B}_{j1}{\bf q}_{j1}
 = \mathbf{0},
\end{equation}
where ${\bf B}_{j1}=\left[ \begin{array}{cccc}
\mathbf{I}_N & \mathbf{A}_{j_1} & \mathbf{A}_{j_2} & \mathbf{0} \\
\mathbf{I}_N & \mathbf{0} & \mathbf{0} & \mathbf{A}_{j_3} \end{array} \right]$ and ${\bf q}_{j1}=\left[\mathbf{a}^T_{j1},\mathbf{q}^{1T}_{j1},\mathbf{q}^{2T}_{j1},\mathbf{q}^{3T}_{j1} \right]^T$.
Similarly, for $\mathbf{a}_{j2} \in \mathcal{S}_{j2}$, we have
\begin{equation}\label{Eq11}
{\bf B}_{j2}{\bf q}_{j2}
 = \mathbf{0},
\end{equation}
where ${\bf B}_{j2}=\left[ \begin{array}{cccc}
\mathbf{I}_N & \mathbf{A}_{j_1} & \mathbf{A}_{j_3} & \mathbf{0} \\
\mathbf{I}_N & \mathbf{0} & \mathbf{0} & \mathbf{A}_{j_2} \end{array} \right]$ and ${\bf q}_{j2}=\left[\mathbf{a}^T_{j2},\mathbf{q}^{1T}_{j2},\mathbf{q}^{2T}_{j2},\mathbf{q}^{3T}_{j2} \right]^T$.
To prove that $\tilde{{\bf A}}$ is of full column rank, it suffices to show that these is no non-zero $\{\mathbf{a}_{ji}\}$, for $j=1,2,3,4$ and $i=1,2$, such that
\begin{equation}\label{Eq2}
\sum^4_{j=1}\sum^{2}_{i=1} \mathbf{a}_{ji} = \mathbf{0}.
\end{equation}
Equivalently, by combining \eqref{Eq1} and \eqref{Eq11}, we need to show that there is no non-zero ${\bf q}$ satisfying
\begin{equation}\label{Eq3}
\left[
  \begin{array}{cccc}
    {\bf B}_1 & {\bf 0} & {\bf 0} & {\bf 0} \\
    {\bf 0} & {\bf B}_2 & {\bf 0} & {\bf 0} \\
    {\bf 0} & {\bf 0} & {\bf B}_3 & {\bf 0} \\
    {\bf 0} & {\bf 0} & {\bf 0} & {\bf B}_4 \\
    {\bf D} & {\bf D} & {\bf D} & {\bf D}
  \end{array}
\right] {\bf q} = {\bf M} {\bf q} = {\bf 0},
\end{equation}
where ${\bf B}_j={\rm diag}({\bf B}_{j1},{\bf B}_{j2})$ for $\forall j$, and ${\bf D}=[{\bf C},{\bf C}]$ with ${\bf C}=[ \mathbf{I}_N, \mathbf{0}, \mathbf{0}, \mathbf{0}]$,
${\bf q} = \left[{\bf q}^T_{11},{\bf q}^T_{12},{\bf q}^T_{21},{\bf q}^T_{22},{\bf q}^T_{31},{\bf q}^T_{32},{\bf q}^T_{41},{\bf q}^T_{42} \right]^T$.
Noting that $\mathbf{A}_i$, for $i = 1, 2, 3,4$, are independent random matrices and $J(t-1)(tM-N)\leq N$, we obtain that $\mathbf{M}\in \mathbb{C}^{(2J(t-1)+1)N\times{J(t-1)(tM+N)}}$ is of full column rank with probability one, which implies that no non-zero $\bf q$ exists to meet \eqref{Eq3}. 

What remains is the case of $N<J(t-1)(tM-N)$. In this case, $\mathcal{S}_1 \oplus \mathcal{S}_2 \oplus \cdots \oplus \mathcal{S}_{J} $ spans the overall signal space of dim-$N$ for sure, which concludes the proof.
\end{proof}

\section{An Example of Symbol Extension}\label{Appendix 3}
Consider a MIMO mRC with $K=3$, $M=2$, and $N=5$. From the discussions in Subsection~\ref{Proof_OneCluster}, as $\frac{1}{3}<\frac{M}{N} < \frac{1}{2}$, we need to construct units following Patterns $1.1$ and $1.3$ to occupy the overall relay's signal space.
The number of independent units with Pattern $1.3$ can be constructed is given by $\frac{3M-N}{2} = \frac{1}{2}$, which is not an integer.
We use the symbol extension technique to avoid the construction of a fractional number of units as follows. Considering two channel uses of \eqref{System-1},
we rewrite the received signal at the relay node as
\begin{eqnarray}\label{Appdix-3-1}
\bar{{\bf y}}_R = \sum^3_{k=1} \bar{{\bf H}}_k \bar{{\bf x}}_k + \bar{{\bf z}}_R,
\end{eqnarray}
where $\bar{{\bf H}}_k = \left[
                           \begin{array}{cc}
                             {\bf H}_k & {\bf 0} \\
                             {\bf 0} & {\bf H}_k \\
                           \end{array}
                         \right] \in \mathbb{C}^{N^\prime \times M^\prime}$ and $\bar{{\bf x}}_k = [{\bf x}_k(2t+1), {\bf x}_k(2t+2)]^T \in \mathbb{C}^{M^\prime \times 1}$ with $N^\prime=2N$ and $M^\prime=2M$. Then the number of units with Pattern $1.3$ is equal to $\frac{3M^\prime-N^\prime}{2} = 1$ and the proposed beamforming design in Subsection~\ref{Proof_OneCluster} can be applied directly.
Note that the two ${\bf H}_k$ in $\bar{{\bf H}}_k$ can be either the same (implying the channel is invariant) or different from each other (implying the channel is time-varying).

%


\end{document}